%% file: approx-colouring.tex
\documentclass[11pt,a4paper]{scrartcl}
\pdfoutput=1
\usepackage[vmarginratio={1:1},vscale=0.75]{geometry}

\usepackage{microtype}
\usepackage[pdftex]{graphicx}
\usepackage{color}
\usepackage[utf8]{inputenc}
\usepackage{amssymb}
\usepackage[english]{babel}
\usepackage{amsmath}
\usepackage{amsfonts}
\usepackage{amsthm}
\usepackage{url}
\usepackage{hyperref}
\usepackage{textcomp}
\usepackage{float}
\usepackage{enumitem}
\usepackage{tikz}
\usepackage{xspace}
\usetikzlibrary{backgrounds,plothandlers,plotmarks,calc}
\usepackage{mathtools}
\usepackage{xcolor}
\usepackage[numbered]{bookmark}
\usepackage{hyperxmp}
\usepackage{titling}
\usepackage[normalem]{ulem}
\usepackage{tikz-cd}
\usepackage{mathrsfs}

\def\theTitle{Improved hardness for $H$-colourings of $G$-colourable graphs}
\title{\theTitle\thanks{Stanislav \v{Z}ivn\'y was supported by a Royal Society
University Research Fellowship. This project has received funding from the
European Research Council (ERC) under the European Union's Horizon 2020 research
and innovation programme (grant agreement No 714532). The paper reflects only
the authors' views and not the views of the ERC or the European Commission. The
European Union is not liable for any use that may be made of the information
contained therein.}}

\author{
Marcin Wrochna\\
University of Oxford, UK\\
\texttt{marcin.wrochna@cs.ox.ac.uk}
\and
Stanislav \v{Z}ivn\'{y}\\
University of Oxford, UK\\
\texttt{standa.zivny@cs.ox.ac.uk}
}

\date{\small First appeared: July 2019, updated: October 2019}

\makeatletter
\hypersetup{
	pdftitle={\theTitle},
	pdfauthor={Marcin Wrochna, Stanislav \v{Z}ivn\'{y}},
	pdfkeywords={PCSP, H-colouring, approximate colouring, homomorphism},
	colorlinks,
	linkcolor={green!40!black},
	citecolor={blue!50!black},
	urlcolor={blue!70!black},
	pdfdisplaydoctitle=true,
}
\makeatother

\newtheorem{theorem}{Theorem}
\newtheorem*{theorem*}{Theorem}
\numberwithin{theorem}{section} 
\newtheorem{lemma}[theorem]{Lemma}
\newtheorem{corollary}[theorem]{Corollary}
\newtheorem{proposition}[theorem]{Proposition}
\newtheorem*{proposition*}{Proposition}

\newtheorem{observation}[theorem]{Observation}
\theoremstyle{definition}
\newtheorem{remark}[theorem]{Remark}

\newtheorem{definition}[theorem]{Definition}
\newtheorem{conjecture}[theorem]{Conjecture}

\newcommand{\ZZ}{\mathbb{Z}}
\newcommand{\NN}{\mathbb{N}}

\newcommand{\RR}{\mathbb{R}}
\newcommand{\eps}{\varepsilon}

\newcommand{\defeq}{\vcentcolon=}

\newcommand{\Cc}{\mathcal{C}}
\newcommand{\Dd}{\mathcal{D}}

\DeclareMathOperator{\PCSP}{PCSP}
\DeclareMathOperator{\Pol}{Pol}
\DeclareMathOperator{\ar}{ar}
\newcommand\B[1]{{\textstyle\binom{#1}{\lfloor #1/2 \rfloor}}}
\DeclareMathOperator{\sub}{sub}
\DeclareMathOperator{\sym}{sym}

\newcommand\Sphere{\mathcal{S}}
\DeclareMathOperator{\Hom}{Hom}
\newcommand\Bx[1]{\mathrm{Box}(#1)}
\newcommand\BBox[1]{\left|\mathrm{Box}(#1)\right|}
\newcommand{\zeq}{\simeq_{\ZZ_2}}
\newcommand{\vL}[1]{{#1}{}\sp{\circ}}
\newcommand{\vR}[1]{{#1}{}\sp{\bullet}}

\usepackage[backend=bibtex,bibencoding=ascii,style=alphabetic,bibstyle=alphabetic,giveninits=true,doi=false,isbn=false,url=false,maxbibnames=99]{biblatex}
\renewbibmacro{in:}{}
\AtEveryBibitem{\clearname{editor}}

\newbibmacro{string+link}[1]{%
	\iffieldundef{ee}
	{\iffieldundef{url}
		{\iffieldundef{doi}
			{\iffieldundef{eprint}
				{#1}
				{\href{http://arxiv.org/abs/\thefield{eprint}}{#1}}}
			{\href{http://dx.doi.org/\thefield{doi}}{#1}}}
		{\href{\thefield{url}}{#1}}}
	{\href{\thefield{ee}}{#1}}} 
\DeclareFieldFormat{title}{\usebibmacro{string+link}{\mkbibemph{#1}}}
\DeclareFieldFormat[article,inproceedings,inbook,incollection,mastersthesis,thesis]{title}{\usebibmacro{string+link}{\mkbibquote{#1}}}

\begin{document}

\maketitle

\vspace*{-1em}	
\begin{abstract}

  We present new results on approximate colourings of graphs and, more
  generally, approximate $H$-colourings and promise constraint satisfaction
  problems.\looseness=-1
  
  First, we show NP-hardness of colouring $k$-colourable graphs with
  $\binom{k}{\lfloor k/2\rfloor}-1$ colours for every $k\geq 4$. This improves
  the result of Bul\'in, Krokhin, and Opr\v{s}al [STOC'19], who gave NP-hardness
  of colouring $k$-colourable graphs with $2k-1$ colours for $k\geq 3$, and the
  result of Huang [APPROX-RANDOM'13], who gave NP-hardness of
  colouring $k$-colourable graphs with $2^{\Omega(k^{1/3})}$ colours for sufficiently
  large $k$. Thus, for $k\geq 4$, we improve from known linear/sub-exponential gaps
  to exponential gaps.
  
  Second, we show that the topology of the box complex of $H$ 
  alone determines whether $H$-colouring of $G$-colourable graphs is NP-hard for all (non-bipartite, $H$-colourable)~$G$.  
  This formalises the topological intuition behind the result of Krokhin and Opr\v{s}al [FOCS'19] that
  $3$-colouring of $G$-colourable graphs is NP-hard for all ($3$-colourable,
  non-bipartite) $G$.
  We use this technique to establish
  NP-hardness of $H$-colouring of $G$-colourable graphs for $H$ that include but
  go beyond $K_3$, including square-free graphs and circular cliques (leaving $K_4$ and larger cliques open).

  Underlying all of our proofs is a very general observation that adjoint
  functors give reductions between promise constraint satisfaction problems.

\end{abstract}

\section{Introduction}

Graph colouring is one of the most fundamental and studied problems in
combinatorics and computer science. A graph $G$ is called $k$-colourable if
there is an assignment of colours $\{1,2,\ldots,k\}$ to the vertices of $G$ so
that any two adjacent vertices are assigned different colours. The chromatic
number of $G$, denoted by $\chi(G)$, is the smallest integer $k$ for which $G$ is
$k$-colourable. Deciding whether $\chi(G)\leq k$ appeared on Karp's original
list of 21 NP-complete problems~\cite{Karp72:reducibility}, and is NP-hard for
every $k\geq 3$. In particular, it is NP-hard to decide whether $\chi(G)\leq 3$
or $\chi(G)>3$. Put differently (thanks to self-reducibility of graph
colouring), it is NP-hard to find a $3$-colouring of $G$ even if $G$ is promised
to be $3$-colourable. 

In the \emph{approximate graph colouring} problem, we are allowed to use more
colours than needed. For instance, given a $3$-colourable graph $G$ on $n$
vertices, can we find a colouring of $G$ using significantly fewer than $n$ colours?
On the positive side, the currently best polynomial-time algorithm of
Kawarabayashi and Thorup~\cite{Kawarabayashi17:jacm} finds a colouring using
$O(n^{0.19996})$ colours. Their work continues a long line of research and is
based on a semidefinite relaxation. On the negative side, it is believed that
finding a $c$-colouring of a $k$-colourable graph is NP-hard for all constants
$3\leq k\leq c$. Already in this regime (let alone for non-constant $c$) our
understanding remains rather limited, despite lots of work and the development
of complex techniques, as we will survey in Section~\ref{sec:related}.

A natural and studied generalisation of graph colourings is that of graph
homomorphisms and, more generally, constraint satisfaction
problems~\cite{Hell08:survey}.

Given two graphs $G$ and $H$, a map $h:V(G)\to V(H)$ is a \emph{homomorphism}
from $G$ to $H$ if $h$ preserves edges; that is, if $\{h(u),h(v)\}\in E(H)$
whenever $\{u,v\}\in E(G)$~\cite{hell2004homomorphism_book}. A celebrated result
of Hell and Ne\v{s}et\v{r}il established a dichotomy for the homomorphism
problem with a fixed target graph~$H$, also known as the \emph{$H$-colouring
problem}: deciding whether an input graph $G$ has a homomorphism to $H$ is
solvable in polynomial time if $H$ is bipartite or if $H$ has a loop; for all
other $H$ this problem is NP-hard~\cite{HellN90}. Note that the $H$-colouring
problem for $H=K_k$, the complete graph on $k$ vertices, is precisely the graph
colouring problem with $k$ colours.

The constraint satisfaction problem (CSP) is a generalisation of the graph
homomorphism problem from graphs to arbitrary relational structures. One type of
CSP that has attracted a lot of attention is the one with a fixed target
structure, also known as the \emph{non-uniform} CSP; see, e.g., the work of
Jeavons, Cohen, and Gyssens~\cite{Jeavons97:jacm},
Bulatov~\cite{Bulatov06:3-elementJACM,Bulatov11:conservative}, and Barto and
Kozik~\cite{Barto14:jacm,Barto16:sicomp}. Following the above mentioned
dichotomy of Hell and Ne\v{s}et\v{r}il for the $H$-colouring~\cite{HellN90} and
a dichotomy result of Schaefer for Boolean CSPs~\cite{Schaefer78:complexity},
Feder and Vardi famously conjectured a dichotomy for all non-uniform
CSPs~\cite{Feder98:monotone}. The Feder-Vardi conjecture was recently confirmed 
independently by Bulatov~\cite{Bulatov17:focs} and Zhuk~\cite{Zhuk17:focs}. In
fact, both proofs establish the so-called ``algebraic dichotomy'', conjectured
by Bulatov, Jeavons, and Krokhin~\cite{Bulatov05:classifying}, which delineates
the tractability boundary in algebraic terms. A high-level idea of the
tractability boundary is that of higher-order symmetries, called polymorphisms,
which allow to combine several solutions to a CSP instance into a new solution.
The lack of non-trivial\footnote{We note that projections/dictators are not the
only trivial polymorphims, cf.~\cite[Example~41]{bkw17:survey}.} polymorphisms guarantees NP-hardness, as shown
already in~\cite{Bulatov05:classifying}. The work of Bulatov and Zhuk show that
\emph{any} non-trivial polymorphism guarantees tractability. We refer the
reader to a recent accessible survey by Barto, Krokhin, and Willard on the
algebraic approach to CSPs~\cite{bkw17:survey}.

Given two graphs $G$ an $H$ such that $G$ is $H$-colourable (i.e., there is
a homomorphism from $G$ to~$H$), the \emph{promise constraint satisfaction
problem} parametrised by $G$ and $H$, denoted by $\PCSP(G,H)$, is the following
computational problem: given a $G$-colourable graph, find an $H$-colouring of
this graph.\footnote{What we described is the ``search version'' of PCSPs. In
the ``decision version'', the goal is to say \textsf{YES} if the input graph is
$G$-colourable and \textsf{NO} if the input graph is not $H$-colourable. The decision
PCSP reduces to the search PCSP but they are not known to be equivalent in
general. However, as far as we know, all known positive results are for the
search version, while all known negative results, including the new results from this
paper, are for the decision version.}
More generally, $G$ and $H$ do not have to be graphs but arbitrary relational
structures. Note that if $G=H$ then we obtain the (search version of the)
standard $H$-colouring and constraint satisfaction problem.

PCSPs have been studied as early as in the classic work of Garey and
Johnson~\cite{Garey76:jacm} on approximate graph colouring but a systematic
study originated in the paper of Austrin, Guruswami, and H{\aa}stad~\cite{Austrin17:sicomp}, who studied a promise version of
$(2k+1)$-SAT, called $(2+\epsilon)$-SAT. In a series of
papers~\cite{Brakensiek16:ccc,Brakensiek18:soda,Brakensiek19:soda}, Brakensiek
and Guruswami linked PCSPs to the universal-algebraic methods developed for the
study of non-uniform CSPs~\cite{bkw17:survey}. In particular, the notion of weak
polymorphisms, identified in~\cite{Austrin17:sicomp}, allowed for some ideas 
developed for CSPs to be be used in the context of PCSPs. The algebraic
theory of PCSPs was then lifted to an abstract level by Bul\'in, Krokhin, and
Opr\v{s}al in~\cite{BulinKO18}. Consequently, this theory was used by Ficak,
Kozik, Ol\v{s}\'ak, and Stankiewicz to obtain a dichotomy for symmetric Boolean
PCSPs~\cite{Ficak19:icalp}, thus improving on an earlier result
from~\cite{Brakensiek18:soda}, which gave a dichotomy for symmetric Boolean PCSP
with folding (negations allowed).

\subsection{Prior and related work}
\label{sec:related}

While the NP-hardness of finding a $3$-colouring of a $3$-colourable graph was
obtained by Karp~\cite{Karp72:reducibility} in 1972, the NP-hardness of finding
a $4$-colouring of a $3$-colourable graph was only proved in 2000 by Khanna,
Linial, and Safra~\cite{Khanna00:combinatorica} (see also the work of Guruswami
and Khanna for a different proof~\cite{Guruswami04:sidma}). This result implied
NP-hardness of finding a $(k+2\lfloor k/3\rfloor-1)$-colouring of a
$k$-colourable graph for $k\geq 3$~\cite{Khanna00:combinatorica}. Early work of
Garey and Johnson established NP-hardness of finding a $(2k-5)$-colouring of a
$k$-colourable graph for $k\geq 6$~\cite{Garey76:jacm}. In 2016, Brakensiek and
Guruswami proved NP-hardness of a $(2k-2)$-colouring of a $k$-colourable
graph for $k\geq 3$~\cite{Brakensiek16:ccc}. Only very recently, Bul\'in,
Krokhin, and Opr\v{s}al showed that finding a  $5$-colouring of a $3$-colourable
graph, and more generally, finding a $(2k-1)$-colouring of a $k$-colourable
graph for any $k\geq 3$, is NP-hard~\cite{BulinKO18}. 

In 2001, Khot gave an asymptotic result -- he showed that for sufficiently large $k$, finding a $k^{\frac{1}{25}(\log k)}$-colouring of a
$k$-colourable graph is NP-hard~\cite{Khot01}. In 2013, Huang improved the gap.
For sufficiently large $k$, he showed that 
finding a $2^{\Omega(k^{1/3})}$-colouring of a $k$-colourable graph is
NP-hard~\cite{Huang13}.

The NP-hardness of colouring ($k$-colourable graphs) with $(2k-1)$ colours for
$k\geq 3$ from~\cite{BulinKO18} and with $2^{\Omega(k^{1/3})}$ colours for sufficiently
large $k$ from~\cite{Huang13} constitute the currently strongest known
NP-hardness results for approximate graph colouring.

Under stronger assumptions (Khot's 2-to-1 Conjecture~\cite{Khot02stoc} for
$k\geq 4$ and its non-standard variant for $k=3$), Dinur, Mossel, and Regev
showed that finding a $c$-colouring of a $k$-colourable graph is NP-hard for
all constants $3\leq k\leq c$~\cite{Dinur09:sicomp}
A variant of Khot's 2-to-1 Conjecture with imperfect completeness has recently been proved~\cite{DinurKKMS18,KhotMS18},
which implies hardness for approximate colouring variants where most but not all of the graph is guaranteed to be $k$-colourable.

Hypergraphs colourings, a special case of PCSPs, is another line of work
intensively studied. A $k$-colouring of a hypergraph is an assignment of colours
$\{1,2,\ldots,k\}$ to its vertices that leaves no hyperedge monochromatic. Dinur,
Regev, and Smyth showed that for any constants $2\leq  k\leq c$, it is NP-hard
to find a $c$-colouring of given $3$-uniform $k$-colourable
hypergraph~\cite{Dinur05:combinatorica}. Other notions of colourings (such as
different types of rainbow colourings) for
hypergraphs were studied by  Brakensiek and
Guruswami~\cite{Brakensiek16:ccc,Brakensiek17:approx}, Guruswami and
Lee~\cite{Guruswami18:combinatorica}, and Austrin, Bhangale, and
Potukuchi~\cite{Austrin18:arxiv}.

Some results are also known for colourings with a super-constant number of colours.
For graphs, conditional hardness was obtained by Dinur and
Shinkar~\cite{Dinur10:approx}. For hypergraphs, NP-hardness results were obtained
in recent work of Bhangale~\cite{Bhangale18:icalp} and Austrin, Bhangale, and
Potukuchi~\cite{Austrin19:arxiv}.

\section{Results}

For two graphs or digraphs $G$, $H$, we write $G \to H$ if there exists a
homomorphism from $G$ to $H$.\footnote{In this paper, we allow graphs to have
loops: the existence of homomorphisms for such graphs is trivial, but this
allows us to make statements about graph constructions that will work without
exceptions.} We are interested in the following computational problem.

\begin{definition}
  Fix two graphs $G$ and $H$ with $G\to H$. The (decision variant of the) $\PCSP(G,H)$ is, given an
  input graph $I$, output $\textsf{YES}$ if $I\to G$, and $\textsf{NO}$ if $I\not\to H$.
\end{definition}

To state our results it will be convenient to use the following definition.

\begin{definition}
	A graph $H$ is \emph{left-hard} if for every non-bipartite graph $G$ with $G \to H$, $\PCSP(G,H)$ is NP-hard.
	A graph $G$ is \emph{right-hard} if for every loop-less graph $H$ with $G \to H$, $\PCSP(G,H)$ is NP-hard.
\end{definition}
If $G \to G'$ and $H' \to H$, then $\PCSP(G,H)$ trivially reduces to  $\PCSP(G',H')$
(this is called \emph{homomorphic relaxation}~\cite{BulinKO18}; intuitively, {increasing} the promise gap makes the problem {easier}).
Therefore, if $H$ is a left-hard graph, then all graphs left of $H$ (that is,
$H'$ such that $H' \to H$) are trivially left-hard.\footnote{Note that by our
definition, bipartite graphs are vacuously left-hard.}
If $G$ is right-hard, then all graphs right of $G$ are right-hard.

For the same reason, since every non-bipartite graph admits a homomorphism from
an odd cycle, to show that $H$ is left-hard it suffices to show that
$\PCSP(C_n,H)$ is NP-hard for arbitrarily large odd $n$, where $C_n$ denotes the
cycle on $n$ vertices.
Dually, since every loop-less graph admits a homomorphism to a clique, to show that $G$ is right-hard it suffices to show that $\PCSP(G,K_k)$ is NP-hard for arbitrarily large $k$.

It is conjectured that all non-trivial PCSPs for (undirected) graphs are NP-hard, greatly extending Hell and  Ne\v{s}et\v{r}il's theorem:

\begin{conjecture}[Brakensiek and Guruswami~\cite{Brakensiek18:soda}]\label{conj:main}
	$\PCSP(G,H)$ is NP-hard for every non-bipartite loop-less $G,H$.
	Equivalently, every loop-less graph is left-hard.
	Equivalently, every non-bipartite graph is right-hard.
\end{conjecture}

In addition to the results on classical colourings discussed above (the case where $G$ and $H$ are cliques),
the following result was recently obtained in a novel application of topological ideas.
\begin{theorem}[Krokhin and Opr\v{s}al~\cite{KrokhinO19}]\label{thm:K3}
	$K_3$ is left-hard.
\end{theorem}
	
\subsection{Improved hardness of classical colouring}

In Section~\ref{sec:arc}, we focus on right-hardness. We use a simple construction called the \emph{arc digraph} or \emph{line digraph}, which decreases the chromatic number of a graph in a controlled way.
The construction allows to conclude the following, in a surprisingly simple way:

\begin{proposition}\label{prop:right-hard}
	There exists a right-hard graph if and only if $K_4$ is right-hard.%
\footnote{Jakub Opr\v{s}al and Andrei Krokhin realised that in this Proposition, 4 can be improved to 3 by using the fact that $\delta(\delta(K_4))$ is 3-colourable, as proved by Rorabaugh, Tardif, Wehlau, and Zaguia~\cite{RTWZ16}. Details will appear in a future journal version.}	
\end{proposition}

More concretely, we show in particular that $\PCSP(K_6,K_{2^k})$ log-space
reduces to $\PCSP(K_4,K_k)$, for all $k \geq 4$. This contrasts
with~\cite[Proposition~10.3]{BartoBKO19},\footnote{\cite{BartoBKO19} is a full
version of~\cite{BulinKO18}. Proposition~10.3 in~\cite{BartoBKO19} is Proposition 5.31 in the
previous two versions of~\cite{BartoBKO19}.} where it is shown to be impossible to
obtain such a reduction with \emph{minion homomorphisms}: an algebraic
reduction, described briefly in Section~\ref{subsec:category}, central to the
framework of~\cite{BulinKO18,BartoBKO19} (in particular, there exists a $k$ such that
$\PCSP(K_4,K_k)$ admits no minion homomorphism to any $\PCSP(K_{n'},K_{k'})$ for
$4 < n' \leq k'$).

Furthermore, we strengthen the best known asymptotic hardness: Huang~\cite{Huang13} showed that for all sufficiently large $n$, $\PCSP(K_n, K_{2^{n^{1/3}}})$ is NP-hard.
We improve this in two ways, using Huang's result as a black-box.
First, we improve the asymptotics from sub-exponential $2^{n^{1/3}}$ to single-exponential  $\B{n} \sim \frac{2^n}{\sqrt{\pi n/2}}$.
Second, we show the claim holds for $n$ as low as $4$.

\newcommand{\ThmAsymp}{
	For all $n \geq 4$, $\PCSP(K_n, K_{\B{n}-1})$ is NP-hard.
}
\begin{theorem}[\textbf{Main Result \#1}]\label{thm:asymp}
  \ThmAsymp
\end{theorem}

In comparison, the previous best result relevant for all integers $n$ was proved
by Bul\'in, Krokhin, and Opr\v{s}al~\cite{BulinKO18}: $\PCSP(K_n, K_{2n-1})$ is NP-hard for all $n\geq 3$.
For $n=3$ we are unable to obtain any results; for $n=4$ the new bound $\B{n}-1=5$ is worse than $2n-1=7$, while for $n=5$ the two bounds coincide at~9.
However, already for $n=6$ we improve the bound from $2n-1=11$ to $\B{n}-1=19$.

\subsection{Left-hardness and topology}

In Section~\ref{sec:functors}, we focus on left-hardness.
The main idea behind Krokhin and Opr\v{s}al's~\cite{KrokhinO19} proof that $K_3$ is left-hard is simple to state.
To prove that $\PCSP(C_n,H)$ is NP-hard for all odd $n$, the algebraic framework of~\cite{BulinKO18} shows that it is sufficient to establish certain properties of \emph{polymorphisms}: homomorphisms $f \colon C_n^{L} \to H$ for $L \in \NN$ (where $G^L=G \times \dots \times G$ is the $L$-fold tensor product\footnote{
	The \emph{tensor} (or \emph{categorical}) \emph{product} $G \times H$ of graphs $G,H$ has pairs $(g,h) \in V(G) \times V(H)$ as vertices and
  $(g,h)$ is adjacent to $(g',h')$ whenever $g$ is adjacent to $g'$ (in $G$) and
  $h$ is adjacent to $h'$ (in $H$).
}).
For large $n$ the graph $C_n^{L}$ looks like an $L$-torus: an $L$-fold product of circles, so the pertinent information about $f$ seems to be subsumed by its topological properties (such as \emph{winding numbers}, when $H$ is a cycle).
We refer to~\cite{KrokhinO19} for further details, but this general principle
applies to any $H$ and in fact we prove (in Theorem~\ref{thm:topoMain} below)
that whether $H$ is left-hard or not depends \emph{only} on its topology.\looseness=-1

The topology we associate with a graph is its \emph{box complex}.
See Appendix~\ref{app:topo} for formal definitions and statements.
Intuitively, the box complex $\BBox{H}$ is a topological space built from $H$ by taking the tensor product $H \times K_2$ and then gluing faces to each four-cycle and more generally, gluing higher-dimensional faces to complete bipartite subgraphs.
The added faces ensure that the box complex of a product of graphs is the same as the product space of their box complexes:
thanks to this, $\BBox{C_n^{L}}$ is indeed equivalent to the $L$-torus.
The product with $K_2$ equips the box complex with a symmetry that swaps the two sides of $H \times K_2$.
This make the resulting space a $\ZZ_2$-space: a topological space together with a continuous involution from the space to itself, which we denote simply as $-$.
A \emph{$\ZZ_2$-map} between two $\ZZ_2$-spaces is a continuous function which preserves this symmetry: $f(-x)=-f(x)$.
This allows to concisely state that a given map is ``non-trivial'' (in contrast, there is always \emph{some} continuous function from one space to another: just map everything to a single point).
The main use of the box complex is then the statement that every graph homomorphism $G \to H$ induces a $\ZZ_2$-map from $\BBox{G}$ to $\BBox{H}$.
Graph homomorphisms can thus be studied with tools from algebraic topology.

The classical example of this is an application of the Borsuk-Ulam theorem:
there is no $\ZZ_2$-map from $\Sphere^n$ to $\Sphere^m$ for $n > m$,
where $\Sphere^n$ denotes the $n$-dimensional sphere with antipodal symmetry.
Hence if $G$ and $H$ are graphs such that $\BBox{G}$ and $\BBox{H}$ are equivalent to $\Sphere^n$ and $\Sphere^m$, respectively,
then there can be no graph homomorphism $G\to H$.
See Figure~\ref{fig:box}.

This is essentially the idea in Lov\'{a}sz' proof~\cite{Lovasz78} of Kneser's conjecture that the chromatic number of Kneser graphs $KG(n,k)$ is $n-2k+2$.
In the language of box complexes, the proof amounts to showing that the box complex of a clique $K_c$ is equivalent to $\Sphere^{c-2}$,  while the box complex of a Kneser graph contains $\Sphere^{n-2k}$.
We refer to~\cite{matousek2008using} for an in-depth, yet accessible reference.

We show that the left-hardness of a graph depends only on the topology of its box complex (in fact, it is only important what $\ZZ_2$-maps it admits, which is significantly coarser than $\ZZ_2$-homotopy equivalence):

\begin{theorem}[\textbf{Main Result \#2}]\label{thm:topoMain}
	If $H$ is left-hard and $H'$ is a graph such that $\BBox{H'}$ admits a $\ZZ_{2}$-map to $\BBox{H}$, then $H'$ is left-hard.
\end{theorem}

\begin{figure}[t!]
	\centering
	\makebox[\textwidth][c]{
		\input{figBox.tex}
	}
	\caption{The box complex of $K_4$ is the hollow cube (informally speaking; the drawing skips some irrelevant faces). It is equivalent ($\ZZ_2$-homotopy equivalent) to the sphere.
		The box complex of the circular clique $K_{7/2}$ is equivalent to the circle.
		Thus there cannot be a homomorphism from $K_4$ to $K_{7/2}$ (of course in this case it is easier to show this directly).}
	\label{fig:box}
\end{figure}
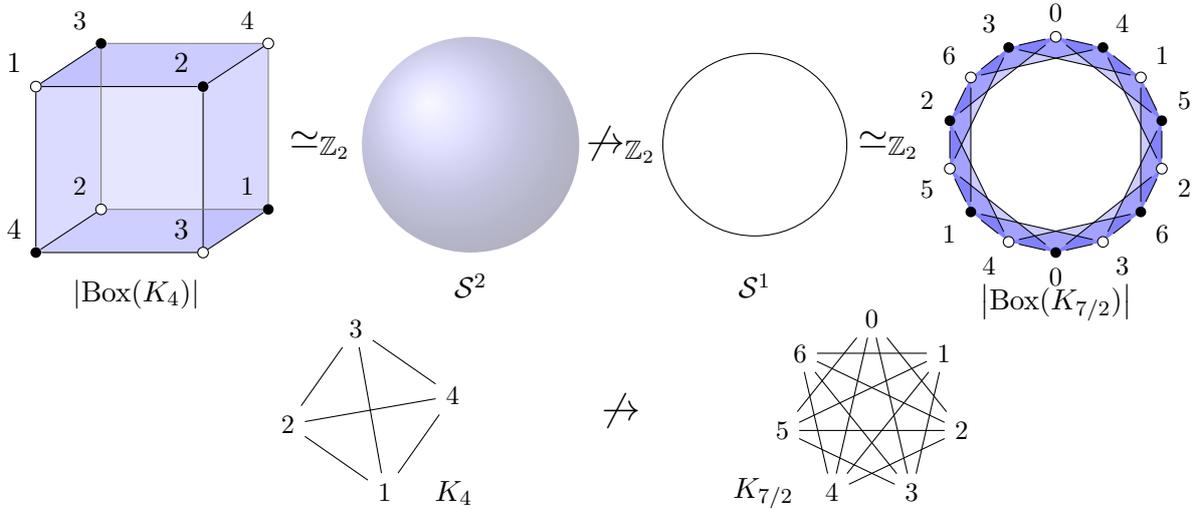

Using Krokhin and Opr\v{s}al's result that $K_3$ is left-hard (Theorem~\ref{thm:K3}), since $\BBox{K_3}$ is the circle $\Sphere^1$ (up to $\ZZ_2$-homotopy equivalence), we immediately obtain the following:
\begin{corollary}\label{cor:S1}
	Every graph $H$ for which $\BBox{H}$ admits a $\ZZ_2$-map to $\mathcal{S}^1$ is left-hard.
\end{corollary}

Two examples of such graphs (other than 3-colourable graphs) are loop-less square-free graphs and circular cliques $K_{p/q}$ with $2<\frac{p}{q}<4$ (see Lemma~\ref{lem:boxes} for proofs), which we introduce next.
\emph{Square-free graphs} are graphs with no cycle of length exactly 4.
In particular, this includes all graphs of girth at least 5 and hence graphs of arbitrarily high chromatic number (but incomparable to $K_4$ and larger cliques, in terms of the homomorphism $\to$ relation).
The \emph{circular clique} $K_{p/q}$ (for $p,q\in \NN, \frac{p}{q}>2$) is the graph with vertex set $\ZZ_p$ and an edge from $i$ to every integer at least $q$ apart: $i+q, i+q+1, \dots, i+p-q$.
They generalise cliques $K_n = K_{n/1}$ and odd cycles $C_{2n+1} \simeq K_{(2k+1)/k}$.
Their basic property is that $K_{p/q} \to K_{p'/q'}$ if and only if $\frac{p}{q} \leq \frac{p'}{q'}$.
Thus circular cliques refine the chain of cliques and odd cycles, corresponding to rational numbers between integers.
For example:
$$ \dots \to C_7 \to C_5 \to C_3 = K_3 \to K_{7/2} \to K_4 \to K_{9/2} \to K_5 \to \dots $$
The \emph{circular chromatic number} $\chi_c(G)$ is the infimum over $\frac{p}{q}$ such that $G \to K_{p/q}$.
Therefore:

\begin{corollary}
	For every $2<r \leq r'<4$, it is NP-hard to distinguish graphs $G$ with $\chi_c(G) \leq r$ from those with $\chi_c(G) > r'$.
\end{corollary}

In this sense, we conclude that $K_{4-\eps}$ is left-hard, thus extending the result for $K_3$.
However, the closeness to $K_4$ is only deceptive and no conclusions on 4-colourings follow.
For $K_4$, since the box complex is equivalent to the standard 2-dimensional sphere, we can at least conclude that to prove left-hardness of $K_4$ it would be enough to prove left-hardness of any other graph with the same topology:
these include all non-bipartite quadrangulations of the projective plane, in
particular the Gr\"{o}tzsch graph, 4-chromatic generalised Mycielskians, and
4-chromatic Schrijver graphs~\cite{matousek2008using,BjornerL03}.
In this sense, the exact geometry of $K_4$ is irrelevant.
However, the fact that it is a finite graph, with only finitely many possible maps from $C_n^L$ for any fixed $n,L$ should still be relevant, as it is for $K_3$.
It is also quite probable that any proof for a ``spherical'' graph would apply just as well to $K_4$, where the proof could be just notationally much simpler.

\bigskip

Finally, in Appendix~\ref{app:topo} we rephrase Krokhin and Opr\v{s}al's~\cite{KrokhinO19} proof of Theorem~\ref{thm:K3} in terms of the box complex.
In particular, left-hardness of $K_3$ follows from some general principles and the fact that $\BBox{K_3}$ is a circle.
The proof also extends to all graphs $H$ such that $\BBox{H}$ admits a $\ZZ_2$-map to $\Sphere^1$,
giving an independent, self-contained proof of Corollary~\ref{cor:S1} (and Theorem~\ref{thm:K3}  in particular).

The general principle is that a homomorphism $C_n^L \to H$ induces a $\ZZ_2$-map $(\Sphere^1)^L \to \BBox{H}$, in a way that preserves minors (identifications within the $L$ variables) and automorphisms. (In the language of category theory, the box complex is a functor from the category of graphs to that of $\ZZ_2$-spaces, and the functor preserves products).
In turn, the $\ZZ_2$-map induces a group homomorphism between the fundamental group of $(\Sphere^1)^L$, which is just $\ZZ^L$, and that of $\BBox{H}$.
This is essentially the map $\ZZ^L \to \ZZ$ obtained in~\cite{KrokhinO19}.
While this rephrasing requires a bit more technical definitions, the main advantage is that it allows to replace a tedious combinatorial argument (about winding numbers preserving minors) with straightforward statements about preserving products.

\subsection{Methodology -- adjoint functors}

While the proof of the first main result is given elementarily in Section~\ref{sec:arc}, it fits together with the second main result in a much more general pattern.
The underlying principle is that pairs of graph constructions satisfying a simple duality condition give reductions between PCSPs.
To introduce them, let us consider a concrete example.
For a graph $G$ and an odd integer $k$, $\Lambda_k G$ is the graph obtained by subdividing each edge into a path of $k$ edges;
$\Gamma_k G$ is the graph obtained by taking the $k$-th power of the adjacency matrix (with zeroes on the diagonal); equivalently, the vertex set remains unchanged and two vertices are adjacent if and only if there is a walk of length exactly $k$ in $G$.
(For example $\Gamma_3 G$ has loops if $G$ has triangles).

We say a graph construction $\Lambda$ (a function from graphs to graphs) is a \emph{thin (graph) functor}
if $G\to H$ implies $\Lambda G \to \Lambda H$ (for all $G,H$).
A pair of thin functors $(\Lambda,\Gamma)$ is a \emph{thin adjoint pair} if
\begin{center}
	$\Lambda G \to H$ if and only if $G \to \Gamma H$.
\end{center}
We call $\Lambda$ the \emph{left adjoint} of $\Gamma$ and $\Gamma$ the \emph{right adjoint} of $\Lambda$.

For all odd $k$, $(\Lambda_k,\Gamma_k)$ are a thin adjoint pair.
For example, since $\Gamma_3 C_5 = K_5$, we have $G \to K_5$ if and only if $\Lambda_k G \to C_5$.
This is a basic reduction that shows the NP-hardness of $C_5$-colouring; in
fact adjointness of various graph construction is the principal tool behind the
original proof of Hell and Ne\v{s}et\v{r}il's theorem (characterising the complexity of
$H$-colouring)~\cite{HellN90}.

In category theory, there is a stronger and more technical notion of (non-thin) \emph{functors} and \emph{adjoint pairs}.
A thin graph functor is in fact a functor in the \emph{thin category of graphs}, that is, the category whose objects are graphs,
and with at most one morphism from one graph to another, indicating whether a homomorphism exists or not.
In other words, we are only interested in the existence of homomorphisms, and not in their identity and how they compose.
Equivalently, we look only at the preorder of graphs by the $G \to H$ relation (we can also make this a poset by considering graphs up to homomorphic equivalence).
In order-theoretic language, thin functors are just order-preserving maps, while thin adjoint functors are known as Galois connections.
We prefer the categorical language as most of the constructions we consider are in fact functors (in the non-thin category of graphs), which is important for connections to the algebraic framework of~\cite{BulinKO18}, as we discuss in Section~\ref{subsec:category}.
While unnecessary for our main results, we believe it may be important to understand these deeper connections to resolve the conjectures completely.

Thin adjoint functors give us a way to reduce one PCSP to another. We say that a graph functor $\Gamma$ is log-space computable if, given a graph $G$, $\Gamma G$ can be computed in logarithmic space in the size of $G$.

\begin{observation}\label{obs:adj1}
	Let $\Lambda,\Gamma$ be thin adjoint graph functors and let $\Lambda$ be log-space computable.
	Then $\PCSP(G, \Gamma H)$ reduces to $\PCSP(\Lambda G, H)$
	in log-space, for all graphs $G,H$.\looseness=-1
\end{observation}
\begin{proof}
	Let $F$ be an instance of $\PCSP(G, \Gamma H)$.
	Then $\Lambda F$ is an appropriate instance of $\PCSP(\Lambda G, H)$.
	Indeed, if $F \to G$, then $\Lambda F \to \Lambda G$ (because $\Lambda$ is a thin functor).
	If $\Lambda F \to H$, then $F \to \Gamma H$ by adjointness.
\end{proof}

In some cases, a thin functor $\Gamma$ that is a thin right adjoint in a pair $(\Lambda, \Gamma)$ is also a thin left adjoint in a pair $(\Gamma,\Omega)$.
This allows to get a reduction in the opposite direction:

\begin{observation}\label{obs:adj2}
	Let $(\Lambda,\Gamma)$ and $(\Gamma,\Omega)$ be thin adjoint pairs of functors.
	Then $\PCSP(\Gamma G, H)$ and $\PCSP(G, \Omega H)$ are log-space equivalent
	(assuming $\Lambda$ and $\Gamma$ are log-space computable).
\end{observation}	
\begin{proof}
	The previous observation gives a reduction from $\PCSP(G, \Omega H)$ to $\PCSP(\Gamma G, H)$. For the other direction, let $F$ be an instance of $\PCSP(\Gamma G, H)$.
	Then $\Lambda F$ is an appropriate instance of $\PCSP(G, \Omega H)$.
	Indeed, if $F \to \Gamma G$, then $\Lambda F \to G$.
	If $\Lambda F \to \Omega H$, then $F \to \Gamma \Omega H \to H$.
	The last arrow follows from the trivial $\Omega H \to \Omega H$.
\end{proof}

The proofs of Observations~\ref{obs:adj1} and~\ref{obs:adj2} of course extend to digraphs and general relational
structures. Note that the above proofs reduce decision problems; they work just as well for search problems: all the thin adjoint pairs $(\Lambda,\Gamma)$ we consider with $\Lambda$ log-space computable also have the property that a homomorphism $\Lambda F \to H$ can be computed from a homomorphism $F \to \Gamma H$ and vice versa, in space logarithmic in the size of $F$.

As we discuss in Section~\ref{sec:functors}, all of our results follow from reductions that are either trivial (homomorphic relaxations) or instantiations of Observation~\ref{obs:adj1}.
While for the first main result we prefer to first give a direct proof that avoids this formalism (in Section~\ref{sec:arc}),
it will be significantly more convenient for the second main result (in Section~\ref{subsec:secondMainProof}),
where we use a certain right adjoint $\Omega_k$ to the $k$-th~power~$\Gamma_k$.

\subsection{Hedetniemi's conjecture}

Another leitmotif of this paper is the application of various tools developed in research around Hedetniemi's conjecture.
A graph $K$ is \emph{multiplicative} if $G \times H \to K$ implies $G \to K$ or $H \to K$.
The conjecture states that all cliques $K=K_n$ are multiplicative.
Equivalently, $\chi(G \times H) = \min(\chi(G),\chi(H))$;
see~\cite{zhu1998survey,Sauer01,Tardif08:survey} for surveys.
In a very recent breakthrough, Shitov~\cite{Shitov19} proved that the conjecture is false (for large $n$).

The arc digraph construction, which we will use in Section~\ref{sec:arc} to
prove Theorem~\ref{thm:asymp}, was originally used by Poljak and Rödl~\cite{PoljakR81} to show certain asymptotic bounds on chromatic numbers of products.
The functors $\Lambda_k,\Gamma_k,\Omega_k$ were applied by
Tardif~\cite{Tardif05:jctb} to show that colourings to circular cliques $K_{p/q}$ ($2<\frac{p}{q}<4$) satisfy the conjecture.
Matsushita~\cite{Matsushita17} used the box complex to show that Hedetniemi's conjecture would imply an analogous conjecture in topology.
This was independently proved by the first author~\cite{Wrochna17b} using $\Omega_k$ functors, while the box complex was used to show that square-free graphs are multiplicative~\cite{Wrochna17}.
See~\cite{FoniokT17} for a survey on applications of adjoint functors to the conjecture.

The refutation of Hedetniemi's conjecture  and the fact that methods for proving the multiplicativity of $K_3$ extend to $K_{4-\eps}$ and square-free graphs, but fail to extend to $K_4$, might suggest that the Conjecture~\ref{conj:main} is doomed to the same fate.
However, it now seems clear that proving multiplicativity requires more than just topology~\cite{TardifW18}: known methods do not even extend to all graphs $H$ such that $\BBox{H}$ is a circle.
This contrasts with Theorem~\ref{thm:topoMain}: topological tools work much more gracefully in the setting of PCSPs.

\section{The arc digraph construction}\label{sec:arc}
Let $D$ be a digraph. The \emph{arc digraph} (or \emph{line digraph}) of $D$, denoted $\delta D$ , is the digraph whose vertices are arcs (directed edges) of $D$ and whose arcs are pairs of the form $((u,v),(v,w))$.
We think of undirected graphs as symmetric relations: digraphs in which for every arc $(u,v)$ there is an arc $(v,u)$.
So for an undirected graph $G$, $\delta(G)$ has $2|E(G)|$ vertices and is a directed graph:
the directions will not be important in this section, but will be in
Section~\ref{subsec:otherExamples}.
The chromatic number of a digraph is the chromatic number of the underlying undirected graph (obtained by symmetrising each arc; so $\chi(D) \leq n$ if and only if $D \to K_n$).

The crucial property of the arc digraph construction is that it decreases the chromatic number in a controlled way (even though it is computable in log-space!).
We include a short proof for completeness. We denote by $[n]$ the set
$\{1,2,\ldots,n\}$.

\begin{lemma}[Harner and Entringer~\cite{HarnerE72}]\label{lem:approxPoljakRodl}
	For any graph $G$:
	\begin{itemize}
		\item if $\chi(\delta(G)) \leq n$, then $\chi(G) \leq 2^n$;
		\item if $\chi(G) \leq \binom{n}{\lfloor n/2\rfloor}$, then $\chi(\delta(G)) \leq n$.
	\end{itemize}
\end{lemma}
\begin{proof}
	Suppose $\delta G$ has an $n$-colouring.
	Recall that we think of $G$ as a digraph with two arcs $(u,v)$ and $(v,u)$ for each edge $\{u,v\} \in E(G)$; thus $\delta G$ contains two vertices $(u,v)$ and $(v,u)$, as well as (by definition of $\delta$) two arcs from one pair to the other.
	In particular, an $n$-colouring of $\delta G$ gives distinct colours to $(u,v)$ and $(v,u)$.
	Define a $2^n$-colouring $\phi$ of $G$ by assigning to each vertex $v$ the set $\phi(v)$ of colours of incoming arcs.
	For any edge $\{u,v\}$ of $G$, $\phi(v)$ contains the colour $c$ of the arc $(u,v)$.
	Since every arc incoming to $u$ gets a different colour from $(u,v)$, the set $\phi(u)$ does not contain $c$.
	Hence $\phi(u) \neq \phi(v)$, so $\phi$ is a proper colouring.
	
	Suppose $G$ has a $\binom{n}{\lfloor n/2\rfloor}$-colouring $\phi$.
	We interpret colours $\phi(v)$ as $\lfloor n/2\rfloor$-element subsets of $[n]$.
	Define an $n$-colouring of $\delta G$ by assigning to each arc $(u,v)$
	an arbitrary colour in $\phi(u)\setminus \phi(v)$ (the minimum, say).
	Such a colour exists because $\phi(u) \neq \phi(v)$.
	For arcs $(u,v)$, $(v,w)$ clearly $\phi(u)\setminus \phi(v)$ is disjoint from $\phi(v) \setminus \phi(w)$, so this is a proper colouring of $\delta(G)$.\looseness=-1
\end{proof}

The proofs in fact works for digraphs as well.
For graphs, it is not much harder to show an exact correspondence
(we note however that most conclusions only require the above approximate correspondence).
Let us denote $b(n)\defeq \B{n}$.

\begin{lemma}[Poljak and R\"odl~\cite{PoljakR81}]\label{lem:PoljakRodl}
	For a (symmetric) graph $G$, 
  \[\chi(\delta(G)) = \min\{n \mid \chi(G) \leq b(n)\}.\]
	In other words, $\delta G \to K_n$ if and only if $G \to K_{b(n)}$.
\end{lemma}

This immediately gives the following implication for approximate colouring:

\begin{lemma}\label{lem:red}
	$\PCSP(K_{b(n)},K_{b(k)})$ log-space reduces to $\PCSP(K_n,K_k)$, for all $n,k \in \NN$.
\end{lemma}
\begin{proof}
	Let $G$ be an instance of the first problem.
	Then $\delta G$ is a suitable instance of $\PCSP(K_n,K_k)$:
	if $G \to K_{b(n)}$, then $\delta G \to K_n$.
	If $\delta G \to K_k$, then $G \to K_{b(k)}$.
\end{proof}

\begin{remark}
	As a side note, adding a universal vertex gives the following obvious reduction:
	$\PCSP(K_n,K_k)$ log-space reduces to $\PCSP(K_{n+1},K_{k+1})$,
	for $n,k \in \NN$.
\end{remark}

Recall also that if $n \leq n' \leq k' \leq k$, then $\PCSP(K_n,K_k)$ trivially reduces to $\PCSP(K_{n'},K_{k'})$.
One corollary of Lemma~\ref{lem:red} is that if any clique of size at least 4 is right-hard, then all of them are:

\begin{proposition}\label{prop:rightHard}
	For all integers $n,n' \geq 4$,
		$\PCSP(K_n,K_k)$ is NP-hard for all $k\geq n$
		if and only if 
		$\PCSP(K_{n'},K_{k'})$ is NP-hard for all $k' \geq n'$.
\end{proposition}
\begin{proof}
	Let $n \leq n'$.
	For one direction, right-hardness of $K_n$ trivially implies right-hardness of $K_{n'}$.
	
	On the other hand, we claim that if $K_{b(n)}$ is right-hard, then so is $K_n$.
	Indeed, suppose $\PCSP(K_{b(n)},K_k)$ is hard for all $k \geq b(n)$.
	In particular it is hard for all $k$ of the form $k=b(k')$ for an integer $k' \geq n$.
	Hence by Lemma~\ref{lem:red}, $\PCSP(K_n,K_{k'})$ is hard for all $k' \geq n$.
	
	Suppose $K_n$ is not right-hard.
	Then $K_{b(n)}$ is not right-hard, $K_{b(b(n))}$ is not right-hard and so on.
	Since starting with $n\geq 4$, the sequence $b(b(\dots n \dots))$ grows to infinity,
	we conclude that $K_{n''}$ is not right-hard for some $n'' \geq n'$.
	Therefore, trivially $K_{n'}$ is not right-hard.
\end{proof}
In other words if any loop-less graph $H$ is right-hard, then trivially some
large enough clique $K_{\chi(H)}$ is right-hard; by the above, $K_4$ and all graphs right of it are right-hard.
This proves Proposition~\ref{prop:right-hard}.
The proof fails to extend to $K_3$ because $b(3)=\B{3}$ is not strictly greater than 3.

\bigskip
The other consequence we derive from Lemma~\ref{lem:red} is a strengthening of Huang's result:

\begin{theorem}[Huang~\cite{Huang13}]\label{thm:Huang}
	For all sufficiently large $n$, $\PCSP(K_n, K_{2^{\Omega(n^{1/3})}})$ is NP-hard.
\end{theorem}

{\renewcommand{\thetheorem}{\getrefnumber{thm:asymp}}
	\begin{theorem}[\textbf{Main Result \#1}]
		\ThmAsymp
	\end{theorem}
	\addtocounter{theorem}{-1}
}
We thus improve the asymptotics from sub-exponential $f(n) \defeq 2^{n^{1/3}}$ to single-exponential $b(n) = \B{n} \sim \frac{2^n}{\sqrt{\pi n/2}}$.
The informal idea of the proof is that any $f(n)$ can be improved to $b^{-1}(f(b(n)))$.
Since $b(n)$ is roughly exponential and $b^{-1}(n)$ is roughly logarithmic, starting from a function $f(n)$ of order $\exp^{(i+1)}(\alpha \cdot \log^{(i)}(n))$ with $i$-fold compositions and a constant $\alpha > 0$, such as $f(n)=2^{n^{1/3}} = 2^{2^{\frac{1}{3} \log n}}$ from Huang's hardness,
results~in 
\[b^{-1}(f(b(n))) \approx \log\Big(\exp^{(i+1)}\big(\alpha \cdot \log^{(i)}(\exp(n))\big)\Big) = \exp^{(i)}(\alpha \cdot \log^{(i-1)}(n)),\]
so a similar composition but with $i$ decreased.
In a constant number of steps, this results in a single-exponential function.
In fact using one more step, but without approximating the function $b(n)$, this results in exactly $b(n)-1$.
We note it would not be sufficient to start from a quasi-polynomial $f(n)$, like
$n^{\Theta(\log n)}$ in Khot's~\cite{Khot01} result.

\begin{proof}[Proof of Theorem~\ref{thm:asymp}]
	By Lemma~\ref{lem:red}:
	\begin{center}
	 $\PCSP(K_{b(n)},K_{b(m)})$ log-space reduces to $\PCSP(K_n, K_m)$, for all $n,m \in \NN$.
	\end{center}
	For any $k \in \NN$, let $m = \lfloor \log k\rfloor$ (all logarithms are base-2); then $b(m) \leq 2^m \leq k$,
	hence $\PCSP(K_{b(n)},K_{k})$ trivially reduces to $\PCSP(K_{b(n)},K_{b(m)})$.

	Therefore, composing the two reductions:
	\begin{center}	
	$\PCSP(K_{b(n)},K_{k})$ reduces to $\PCSP(K_{n},K_{\lfloor \log k \rfloor})$, for any $n,k \in \NN$.
	\end{center}
	Starting from Theorem~\ref{thm:Huang} we have a constant $C$ such that:
	\begin{center}
		$\PCSP(K_{n}, K_{2^{\lfloor C \cdot n^{1/3}\rfloor}})$ is NP-hard, for sufficiently large $n$.
	\end{center}	
	Hence, substituting $n=b(k)$:
	\begin{center}
		 $\PCSP(K_{b(k)}, K_{2^{\lfloor C \cdot b(k)^{1/3}\rfloor }})$ is NP-hard, for sufficiently large $k$.
	\end{center}
	Applying the above reduction, since $\lfloor \log 2^{\lfloor C \cdot b(k)^{1/3}\rfloor} \rfloor = \lfloor C \cdot b(k)^{1/3}\rfloor \geq (\frac{2^k}{k})^{1/3}\geq 2^{k/4}$ for sufficiently large $k$, we conclude:
	\begin{center}
		$\PCSP(K_{k}, K_{2^{k/4}})$ is NP-hard, for sufficiently large $k$.
	\end{center}
	We repeat this process to bring the constant further ``down''.
	That is, we substitute $b(k)$ for $k$ and apply the above reduction again.
	Since $\lfloor \log 2^{b(k)/4} \rfloor = \lfloor b(k)/4 \rfloor \geq 2^k/4k$ for sufficiently large $k$, we conclude:
	\begin{center}
	$\PCSP(K_{k}, K_{2^{k}/4k})$ is NP-hard, for sufficiently large $k$.
	\end{center}
	To apply the reduction one more time, notice that for large $k$, $b(k) \geq \frac{3}{2} b(k-1)$ (because $b(2k) = \binom{2k}{k} = \binom{2k-1}{k-1} \frac{2k}{k} = 2 \cdot b(2k-1) \geq \frac{3}{2} b(2k-1)$ and $b(2k+1) = \binom{2k+1}{k} = \binom{2k}{k} \frac{2k+1}{k+1} = b(2k) (2-\frac{1}{k+1})\geq \frac{3}{2} b(2k)$).
	Therefore $\lfloor \log (2^{b(k)}/4b(k)) \rfloor \geq b(k) - \log b(k) \geq \frac{2}{3} b(k) \geq b(k-1)$ for sufficiently large $k$, hence:
	\begin{center}
	$\PCSP(K_{k}, K_{b(k-1)})$ is NP-hard, for sufficiently large $k$.
	\end{center}
	Substituting $b(k)$ for $k$ one last time:
	\begin{center}
	$\PCSP(K_{b(k)}, K_{b(b(k)-1)})$ is NP-hard, for sufficiently large $k$.
	\end{center}
	Composing with Lemma~\ref{lem:red} one last time:
	\begin{center}
	$\PCSP(K_{k}, K_{b(k)-1})$ is NP-hard, for sufficiently large $k$.
	\end{center}
	This concludes the improvement in asymptotics.
	Moreover, one can notice that the requirements on ``sufficiently large $k$'' gets
	relaxed whenever we substitute $b(k)$ for~$k$.
	Formally,
	let $k$ be maximum such that $\PCSP(K_{k}, K_{b(k)-1})$ is not NP-hard.
	Then because of Lemma~\ref{lem:red}, $\PCSP(K_{b(k)}, K_{b(b(k)-1)})$ is not NP-hard,
	and because $b(b(k)-1) \leq b(b(k))-1$, trivially $\PCSP(K_{b(k)}, K_{b(b(k))-1})$ is not NP-hard either.
	That is, $\PCSP(K_n,K_{b(n)-1})$ is not NP-hard for $n=b(k)$.
	By maximality of $k$, $k \geq n$. But $k \geq b(k)$ is only possible when $k < 4$.
	Hence hardness holds for all $k \geq 4$.
\end{proof}

\section{Adjoint functors and topology}\label{sec:functors}
\subsection{\texorpdfstring{Thin functors $\Lambda_k,\Gamma_k,\Omega_k$}{Thin functors}}
\label{subsec:secondMainProof}
Recall that $\Lambda_k$ denotes $k$-subdivision and $\Gamma_k$ denotes the $k$-th power of a graph.
For all odd $k$, they are thin adjoint graph functors:
\begin{center}
	$\Lambda_k G \to H$ if and only if $G \to \Gamma_k H$.
\end{center}

\noindent
More surprisingly, $\Gamma_k$ is itself the thin \emph{left} adjoint of a certain thin functor $\Omega_k$:
\begin{center}
	$\Gamma_k G \to H$ if and only if $G \to \Omega_k H$.
\end{center}
This characterizes $\Omega_k G$ up to homomorphic equivalence.
The exact definition is irrelevant, but we state it for completeness:
for $k=2\ell+1$, the vertices of $\Omega_k$ are tuples $(A_0,\dots,A_\ell)$ of vertex subsets $A_i \subseteq V(G)$ such that $A_0$ contains exactly one vertex.
Two such tuples $(A_0,\dots,A_\ell)$ and $(B_0,\dots,B_\ell)$ are adjacent if
$A_i \subseteq B_{i+1}$, $B_i \subseteq A_{i+1}$ (for $i=0\dots \ell-1$) and $A_\ell$ is fully adjacent to $B_\ell$ (meaning $a$ is adjacent to $b$ in $G$, for $a \in A_k, b \in B_k$).
We note that $\Lambda_k$ and $\Gamma_k$ are log-space computable, for all odd $k$;
however, $\Omega_k$ is not: $\Omega_k G$~is exponentially larger than $G$.
See~\cite{Wrochna17b} for more discussion about the thin functors $\Lambda_k,\Gamma_k,\Omega_k$ and their properties.

Observation~\ref{obs:adj1} tells us that $\PCSP(G, \Omega_k H)$ log-space reduces to $\PCSP(\Gamma_k G, H)$ (in fact, by Observation~\ref{obs:adj2}, they are equivalent).
To give conclusions on left-hardness, we will need to observe only two more facts about the functors $\Lambda_k,\Gamma_k,\Omega_k$.
First, $\Omega_k G \to G$ for all $G$ (it suffices to map $(A_0,\dots,A_{l-1},A_\ell) \in V(\Omega_{2\ell+1} G)$ to the unique vertex in $A_0$).
Second, it is not hard to check that $\Gamma_k \Lambda_k G \to G$ and hence by adjointness $\Lambda_k G \to \Omega_k G$ for all $G$ and odd $k$ (see Lemma 2.3 in~\cite{Wrochna17b}).

\begin{lemma}\label{lem:left-hard}
	For every odd $k$, $\Omega_k H$ is left-hard if and only if $H$ is left-hard.
\end{lemma}
\begin{proof}
	If $H$ is left-hard, then trivially so is $\Omega_k H$ because $\Omega_k H \to H$.
	For the other implication, suppose $\Omega_k H$ is left-hard, that is,
	$\PCSP(G, \Omega_k H)$ is hard for every non-bipartite $G$ such that $G \to \Omega_k H$.
	By Observation~\ref{obs:adj1}, this implies $\PCSP(\Gamma_k G, H)$ is hard.
	Let $G'$ be any non-bipartite graph such that $G' \to H$.
	We want to show that $\PCSP(G', H)$ is hard.
	Observe that $\Omega_k G'$ is non-bipartite, because $\Lambda_k G' \to \Omega_k G'$ and $\Lambda_k$ subdivides each edge of $G'$ an odd number of times.
	Since $\Omega_k G' \to \Omega_k H$,
	using $G := \Omega_k G'$ we conclude that $\PCSP(\Gamma_k \Omega_k G', H)$ is hard.
	Since $\Gamma_k \Omega_k G' \to G'$,
	this implies $\PCSP(G', H)$ is hard.
\end{proof}

As an example, consider the circular clique $K_{7/2}$ (we have $K_3 \to K_{7/2} \to K_4$).
Knowing that $K_3$ is left-hard, one could check that $\Omega_3(K_{7/2})$ is 3-colorable and hence left-hard as well;
the above lemma then allows to conclude that $K_{7/2}$ is left-hard.

What other graphs could one use in place of $K_{7/2}$?
The answer turns out to be topological.
Intuitively, while the operation $\Gamma_k$ gives a ``thicker'' graph,
the operation $\Omega_k$ gives a ``thinner'' one.
In fact, $\Omega_k$ behaves like barycentric subdivision in topology: it preserves the topology of a graph (formally: its box complex is $\ZZ_2$-homotopy equivalent to the original graph's box complex) but refines its geometry.
With increasing $k$, this eventually allows to model any continuous map with a graph homomorphism; in particular:

\begin{theorem}[\cite{Wrochna17b}]\label{thm:approx}
	There exists a $\ZZ_2$-map $\BBox{G} \to_{\ZZ_2} \BBox{H}$ if and only if for some odd $k$, $\Omega_k G \to H$.
\end{theorem}

\noindent
This concludes our second main result:

\begin{proof}[Proof of Theorem~\ref{thm:topoMain}]
	Let $H$ be left-hard and let $H'$ be a graph such that $\BBox{H'}$ admits a $\ZZ_{2}$-map to $\BBox{H}$.
	By Theorem~\ref{thm:approx}, $\Omega_k H' \to H$ for some odd $k$.
	Trivially then, $\Omega_k H'$ is left-hard.
	By Lemma~\ref{lem:left-hard}, $H'$ is left-hard.
\end{proof}

\subsection{Other examples	of adjoint functors}\label{subsec:otherExamples}
The arc construction $\delta$ is also an example of a digraph functor which admits both a thin left adjoint $\delta_L$ and a thin right adjoint $\delta_R$;%
\footnote{For the interested reader: $\delta_L D$ is obtained by making a new arc $(s_v,t_v)$ for each vertex of $D$ and then for each arc $(u,v)$ of $D$, gluing $t_u$ with $s_v$ (which results in many transitive gluings); $\delta_R D$ has a vertex for each pair $S,T \subseteq V(D)$ such that $S \times T \subseteq E(D)$, and an arc from $(S,T)$ to $(S',T')$ iff $T \cap S' \neq \emptyset$.}
this adjointness essentially gives a proof of Lemma~\ref{lem:approxPoljakRodl}, see~\cite[Proposition~3.3]{FoniokT17}.
In fact, Lemma~\ref{lem:red}, and hence all results of Section~\ref{sec:arc}, can be deduced
as instantiations of Observation~\ref{obs:adj1} and homomorphic relaxations as follows.
Let $\sym(D)$ be the symmetric closure of a digraph $D$ and let $\sub(D)$ be the maximal symmetric subgraph of $D$; note $\sub(D) \to D \to \sym(D)$.
Observe that they are thin adjoint functors: $\sym(D) \to D'$ if and only if $D \to \sub(D')$, for all digraphs $D,D'$.\footnote{As Jakub Opr\v{s}al observed, this is in fact the composition of two adjoint pairs: taking $\sym$ and $\sub$ as functors from digraphs to graphs and the inclusion functor $\iota$ from graphs to digraphs, we have $\sym(D) \to G$ iff $D \to \iota(G)$ and $\iota(G) \to D$ iff $G \to \sub(D)$.}
Poljak and R\"odl~\cite{PoljakR81} showed that $\sub(\delta_R(K_{k})) \to K_{b(k)}$ (the $\sub$ is essential here); recall also that $\delta(\sym(K_{b(n)})) \to K_n$.
Therefore, $\PCSP(K_{b(n)},K_{b(k)})$ trivially reduces to $\PCSP(K_{b(n)},\sub(\delta_R(K_{k})))$,
which by Observation~\ref{obs:adj1} log-space reduces to $\PCSP(\delta(\sym(K_{b(n)})),K_k)$, which trivially reduces to $\PCSP(K_n,K_k)$,
proving Lemma~\ref{lem:red}.
From Observation~\ref{obs:adj2} we also have:

\begin{corollary}
	$\PCSP(\delta(G),H)$ is log-space equivalent to $\PCSP(G, \delta_R(H))$, 
	for all digraphs~$G, H$.
\end{corollary}

Another example of a thin adjoint pair (but not triple) of functors is given by products and exponential graphs (see e.g.~\cite{FoniokT13} for definitions):
for any graphs $F,G,H$, we have $F \times G \to H$ if and only if $G \to H^F$.
That is, for any graph $F$, the operations $G \mapsto F \times G$ and $H \mapsto H^F$ are left and right adjoints, respectively.
By Observation~\ref{obs:adj1}:

\begin{corollary}
	$\PCSP(G, H^F)$ reduces to $\PCSP(F \times G, H)$ in log-space.
\end{corollary}

Here $\times$ is the \emph{tensor} (or \emph{categorical}) product,
in particular $G \to H_1 \times H_2$ if and only if $G \to H_1$ and $G \to H_2$.
Nevertheless, a few other products have an associated exponentiation as well.
These and other examples fall into a pattern known as \emph{Pultr functors} -- see~\cite{FoniokT13} for an extended discussion (we note here that \emph{central Pultr functors}, like $\Gamma_k$ or $\delta$, are a kind of pp-interpretation).
Foniok and Tardif~\cite{FoniokT15} studied which digraph functors admit both thin left and right adjoints.

The box complex also admits a left adjoint, though they involve two categories.
More precisely, the functor $G \mapsto \Hom(K_2, G)$ (see definitions in Appendix~\ref{app:topo})
gives a $\ZZ_2$-simplicial complex that is $\ZZ_2$-homotopy equivalent to the box complex.
As proved by Matsushita~\cite{Matsushita17}, it admits a left adjoint $A$ from the category of $\ZZ_2$-simplicial complexes (with $\ZZ_2$-simplicial maps as morphisms) to the category of graphs.

\subsection{Relation to the algebraic framework}\label{subsec:category}

We will need basic concepts from the algebraic approach to (P)CSPs, such as
polymorphisms~\cite{Austrin17:sicomp,Brakensiek18:soda}, minions, and minion
homomorphisms~\cite{BulinKO18}. We shall define them only for graphs as we do
not need them for relational structures. We refer the reader
to~\cite{bkw17:survey,BulinKO18} for more details, examples, and general
definitions.

An $n$-ary \emph{polymorphism} of two graphs $G$ and $H$ is a homomorphism from
$G^n$ to $H$; that is, a map $f:V(G)^n\to V(H)$ such that, for all edges
$(u_1,v_1),\ldots,(u_n,v_n)$ in $G$,
$(f(u_1,\ldots,u_n),f(v_1,\ldots,v_n))$ is an edge in $H$. We denote by
$\Pol(G,H)$ the set of all polymorphisms of $G$ and $H$.

Given an $n$-ary function $f:A^n\to B$, the, say, first coordinate is called
\emph{essential} if there exist $a,a'\in A$ and $\vec{a}\in A^{n-1}$ such that
$f(a,\vec{a})\neq f(a',\vec{a})$; otherwise, the first coordinate is called
\emph{inessential} or \emph{dummy}. Analogously, one defines the $i$-th
coordinate to be (in)essential.
The \emph{essential arity} of $f$ is the number of essential coordinates.

Let $f:A^n\to B$ and $g:A^m\to B$ be $n$-ary and $m$-ary functions,
respectively. We call $f$ a \emph{minor} of $g$ if $f$ can be obtained from $g$ by
identifying variables, permuting variables, and introducing inessential
variables. More formally, $f$ is a minor of $g$ given by a map
$\pi:[m]\to[n]$ if $f(x_1,\ldots,x_n)=g(x_{\pi(1)},\ldots,x_{\pi(m)})$.

A \emph{minion} on a pair of sets $(A,B)$ is a non-empty set of functions (of
possibly different arities) from $A$ to $B$ that is closed under taking minors.
A minion is said to have \emph{bounded essential arity} if there is some $k$
such that every function from the minion has essential arity at most $k$.

Let $\mathscr M $ and $\mathscr N $ be two minions, not necessarily on the same pairs of
sets. A map $\xi:\mathscr M\to\mathscr N$ is called a \emph{minion
homomorphism} if (1) it preserves arities; i.e., maps $n$-ary functions to
$n$-ary functions, for all $n$; and (2) it preserves taking minors; i.e., for
each $\pi:[m]\to[n]$ and each $m$-ary $g\in\mathscr M$, we have
$\xi(g)(x_{\pi(1)},\ldots,x_{\pi(m)})=\xi(g(x_{\pi(1)},\ldots,x_{\pi(m)}))$.
Minion homomorphisms provide an algebraic way to give reductions between PCSPs.

\begin{theorem}[\cite{BulinKO18}]
	If there is a minion homomorphism $\xi \colon \Pol(G_1,H_1) \to \Pol(G_2,H_2)$,
	then $\PCSP(G_2,H_2)$ is log-space reducible to $\PCSP(G_1,H_1)$.
\end{theorem}

The following hardness result is a special case of a result obtained
in~\cite{BulinKO18} via a reduction from Gap Label Cover. It gives an algebraic
tool to prove hardness for PCSPs.

\begin{theorem}[\cite{BulinKO18}]\label{thm:minbndarity}
Let $G$ and $H$ be two graphs with $G\to H$. Assume that there exists a minion
  homomorphism $\xi:\Pol(G,H)\to\mathscr M$ for some minion $\mathscr M$ on a
  pair of (possibly infinite) sets such that $\mathscr M$ has bounded essential
  arity and does not contain a constant function (i.e., a function without
  essential variables). Then $\PCSP(G,H)$ is NP-hard.
\end{theorem}

Our methods do not give minion homomorphisms in general: while Observation~\ref{obs:adj1} gives a reduction from $\PCSP(G,\Gamma H)$ to $\PCSP(\Lambda G, H)$, it does not give a minion homomorphism from which the reduction would follow (from $\Pol(\Lambda G, H)$ to $\Pol(G,\Gamma H)$).
Indeed it cannot, as discussed below Proposition~\ref{prop:right-hard}.
However, adjoint functors in the (non-thin) category of graphs do imply such a minion homomorphism.

In the remainder of this section, we assume knowledge of basic definitions in category theory.
One can define minions in any Cartesian category $\Cc$ (i.e. a category with all finite products), using morphisms of $\Cc$ in place of functions.
For objects $G,H \in \Cc$, $\Pol_{\Cc}(G,H)$ is the minion of morphisms from $G^L$ (the $L$-fold categorical product of $G$) to $H$.
A function $\pi\colon [L] \to [L']$ induces a morphism $\pi_G \colon G^{L'} \to G^L$.
For a graph $G$, it maps $(v_1,\dots,v_{L'})$ to $(v_{\pi(1)},\dots,v_{\pi(L)})$.
In general, it can be defined as the product morphism $\langle p_{\pi(1)},\dots,p_{\pi(L)} \rangle$ of appropriate projections $p_i \colon G^L \to G$.
For a polymorphism $f \colon G^L \to H$, the minor of $f$ by $\pi$ is then simply $f \circ \pi_G \colon G^{L'} \to H$.

For objects $G$ and $H$ of a category, we denote by $\hom(G,H)$ the set of morphisms from $G$ to $H$.

\begin{lemma}\label{lem:cat}
	Let $\Gamma\colon \Cc\to\Dd$ and $\Omega\colon \Dd\to\Cc$ be adjoint functors between Cartesian categories $\Cc,\Dd$.
	Then for all objects $G$ in $\Cc$ and $H$ in $\Dd$,
	there is a minion homomorphism from $\Pol_{\Dd}(\Gamma G, H)$ to $\Pol_{\Cc}(G,\Omega H)$.
	If, moreover, $\Gamma$ preserves products then this is a minion isomorphism.
\end{lemma}
\begin{proof}
	This essentially amounts to checking definitions.
	We have a natural morphism $\psi_L \colon \Gamma(G^L) \to (\Gamma G)^L$
	defined as the product morphism $\langle \Gamma p_1,\dots, \Gamma p_L \rangle$ for projections $p_i \colon G^L \to G$.
	It is natural in the following sense: for every function $\pi\colon [L] \to [L']$, the following diagram commutes:\vspace*{-\baselineskip}
	$$\begin{tikzcd}
		\Gamma(G^{L'}) \ar[d, "\Gamma \pi_G"] \ar[r, "\psi_{L'}"] &
		(\Gamma G)^{L'} \ar[d, "\pi_{\Gamma G}"] \\
		\Gamma(G^L) \ar[r, "\psi_L"] & (\Gamma G)^L
	\end{tikzcd}$$
	\noindent
	Indeed, $\psi_L \circ  \Gamma \pi_G = \pi_{\Gamma G} \circ \psi_{L'}$, because it is the unique morphism whose composition with $p'_i \colon (\Gamma G)^L \to \Gamma G$ is $\Gamma p_{\pi(i)}$ (in other words, it is the product morphism $\langle \Gamma p_{\pi(1)},\dots, \Gamma p_{\pi(L)} \rangle$).

	Let $\Omega$ be a right adjoint of $\Gamma$.
	Let $\Phi_{G^L\kern-.2em,H} \colon  \hom(\Gamma(G^L),H) \to \hom(G^L,\Omega H)$ be the natural isomorphism given by definition of adjunction. Naturality here means that in particular the right square in the following diagram commutes:

	\[\begin{tikzcd}
		\hom((\Gamma G)^L,    H) \ar[d, "-\circ\pi_{\Gamma G}"] \ar[r, "-\circ\psi_L"]               &
		\hom(\Gamma(G^L),    H)  \ar[d, "-\circ\Gamma\pi_G"]    \ar[r, "\Phi_{G^L\kern-.2em,H}"]     &			
		\hom(G^L,   \Omega H)    \ar[d, "-\circ\pi_G"]                                               \\
		\hom((\Gamma G)^{L'}, H)                                \ar[r, "-\circ\psi_{L'}"]            &
		\hom(\Gamma(G^{L'}), H)                                 \ar[r, "\Phi_{G^{L'}\kern-.2em,H}"]  &
		\hom(G^{L'},\Omega H)                                   
	\end{tikzcd}\]

	\noindent
	That is, for $f\colon \Gamma(G^L) \to \Omega H$, we have $\Phi_{G^L\kern-.2em,H}(f) \circ \pi_G = \Phi_{G^{L'}\kern-.2em,H}(f \circ \Gamma \pi_G)$.
	The left square also commutes because of the previously discussed commutation.
	Therefore, we can define a minion homomorphism $\xi \colon \hom((\Gamma G)^L, H) \to \hom(G^L,   \Omega H)$
	as $\xi(f) \defeq \Phi_{G^L\kern-.2em,H}(f \circ \psi_L)$.
	Indeed, $\xi$ preserves minors, because
	$\xi(f \circ \pi_{\Gamma G}) = \xi(f) \circ \pi_G$ as seen on the perimeter of the above diagram.

	If $\Gamma$ preserves products, then $\psi_L$ is an isomorphism.
	Since $\Phi_{G^L\kern-.2em,H}$ is a bijection, this means $\xi$ is a minion isomorphism.
\end{proof}

A basic lemma in category theory says that if a functor $\Gamma$ admits a left adjoint, then it preserves products (indeed, all limits).
So a pair of adjoint pairs $(\Lambda,\Gamma)$, $(\Gamma,\Omega)$ implies a minion isomorphism.
Hence the first part of Lemma~\ref{lem:cat} is analogous to Observation~\ref{obs:adj1},
while the second part is analogous to Observation~\ref{obs:adj2}.
We can also derive the second direction as a corollary to the following lemma.
\begin{lemma}\label{lem:cat2}
	Let $\Gamma\colon \Cc \to \Dd$ be a functor which preserves products.
	Then there is a minion homomorphism $\Pol_\Cc(G,H) \to \Pol_{\Dd}(\Gamma G, \Gamma H)$, for all $G,H \in \Cc$.
\end{lemma}
\begin{proof}	
	Recall from the proof of Lemma~\ref{lem:cat} the following diagram, for $G \in \Cc$, $L,L' \in \NN$, and $\pi\colon [L] \to [L']$:
	$$\begin{tikzcd}
		\Gamma(G^{L'}) \ar[d, "\Gamma \pi_G"] \ar[r, "\psi_{L'}"] &
		(\Gamma G)^{L'} \ar[d, "\pi_{\Gamma G}"] \\
		\Gamma(G^L) \ar[r, "\psi_L"] &
		(\Gamma G)^L
	\end{tikzcd}$$
	Since $\Gamma$ preserves products, $\psi_L$ is an isomorphism,
	so we can define a minion homomorphism $\xi\colon \Pol_\Cc(G,H) \to \Pol_{\Dd}(\Gamma G, \Gamma H)$ as follows: $\xi(f) \defeq \Gamma(f) \circ \psi_{L}^{-1}$, for $f \colon G^L \to H$.
	This preserves minors, because from the diagram's commutation we have:
	\[\xi(f \circ \pi_G) = \Gamma(f \circ \pi_G) \circ \psi_{L'}^{-1} = \Gamma(f) \circ \Gamma(\pi_G) \circ \psi_{L'}^{-1} =
	\Gamma(f) \circ \psi_L^{-1} \circ \pi_{\Gamma G} = \xi(f) \circ \pi_{\Gamma G}.\vspace*{-\baselineskip}\]
\end{proof}

\begin{corollary}\label{cor:cat}
	Let $\Gamma\colon \Cc \to \Dd$ be a functor which preserves products.
	Let $\Omega$ be a thin right adjoint to $\Gamma$.
	Then there is a minion homomorphism $\Pol_\Cc(G,\Omega H) \to \Pol_{\Dd}(\Gamma G, H)$ for all $G \in \Cc, H \in \Dd$.
\end{corollary}
\begin{proof}
	Since $\Gamma$ has a \emph{thin} right adjoint $\Omega$,
	there exists a morphism $\eps_H\colon \Gamma \Omega H \to H$ for all $H$ (we don't need it to be natural in any way).
	Hence we can compose the minion homomorphism $\Pol_\Cc(G,\Omega H) \to \Pol_{\Dd}(\Gamma G, \Gamma \Omega H)$
	from Lemma~\ref{lem:cat2} with the trivial minion homomorphism $\Pol_{\Dd}(\Gamma G, \Gamma \Omega H) \to \Pol_{\Dd}(\Gamma G, H)$ obtained by composing with $\eps_H$.
\end{proof}

If we have adjoint functors in the (non-thin) category of graphs (or multigraphs), then Lemma~\ref{lem:cat} implies a minion homomorphism between the standard polymorphism minions (because a morphism is associated with a function between vertex sets).
One could also apply Lemma~\ref{lem:cat} to the \emph{thin} category of graphs, but the conclusion is then about minions of polymorphisms in that thin category, which is useless, since it does not distinguish between different projections $G^L \to G$.

All the thin functors we have considered are in fact functors in the category of graphs or digraphs: in particular $\Lambda_k,\Gamma_k,\Omega_k,\delta_L,\delta,\delta_R$.
The definitions can also be extended to give functors in the category of multi(di)graphs.
The pairs $(\Lambda_k,\Gamma_k)$ and $(\delta_L,\delta)$ are adjoint pairs in the categories of multi(di)graphs
(this fails in the category of (di)graphs; e.g. the number of homomorphisms $\Lambda_3 G \to H$ is not always equal to the number of homomorphisms $G \to \Gamma_3 H$).
This implies minion homomorphisms $\Pol(\Lambda_k G, H) \to \Pol(G,\Gamma_k H)$ and 
$\Pol(\delta_L G, H) \to \Pol(G,\delta H)$.

In contrast, the pairs $(\Gamma_k,\Omega_k)$ and $(\delta,\delta_R)$ are not adjoint pairs; they are only thin adjoints.
Since $\Gamma_k$ and $\delta$ are right adjoints (of $\Lambda_k$ and $\delta_L$), they preserve products.
Applying Corollary~\ref{cor:cat} hence at least gives minion homomorphisms
$\Pol(G, \Omega_k H) \to \Pol(\Gamma_k G, H)$ and 
$\Pol(G, \delta_R H) \to \Pol(\delta G, H)$.
However, our results would only follow from the opposite direction.
This is impossible to obtain in general: a minion homomorphism
$\Pol(\delta G, H) \xrightarrow{?} \Pol(G, \delta_R H)$
would imply the following minion homomorphism
\[\Pol(K_4, K_k) \to \Pol(\delta K_6, K_k) \xrightarrow{?} \Pol(K_6, \delta_R K_k) \to \Pol(K_6, K_{2^k})\]
(trivially from $\delta K_6 \to K_4$ and $\delta_R K_k \to K_{2^k}$),
which is impossible by~\cite[Proposition~10.3]{BartoBKO19}.
Thus the seemingly technical difference between adjoints and thin adjoints turns out to be crucial.

As proved by Matsushita~\cite{Matsushita17}, the hom complex $\Hom(K_2,-)$ has a left adjoint from the category of $\ZZ_2$-simplicial complexes with $\ZZ_2$-simplicial maps to the category of graphs; the left adjoint preserves products.

\section{Conclusions}
The reduction in Lemma~\ref{lem:red}, on which our first main result relies, does not have a corresponding minion homomorphism.
Given the simplicity of the reduction itself, this contrasts with the success of minion homomorphism in explaining other reductions between promise constraint satisfaction problems.
It is to been seen whether this notion can be extended to a more general relation between polymorphism sets in a way that would imply Lemma~\ref{lem:red}.

The question of whether $K_4$ is left-hard stands open.
In principle, it may be possible to extend the proof in Appendix~\ref{app:topo} using more tools from algebraic topology to analyse $\ZZ_2$-maps $(\Sphere^1)^L \to \Sphere^2$ and deduce an appropriate minion homomorphism.
It could also be interesting to consider how $\delta$ or $\delta_R$ affect the topology of a graph, cliques in particular.

Another direction could be to look at Huang's Theorem~\ref{thm:Huang} not as a black-box: could constructions like $\delta$ be useful to say something directly about PCPs?

\appendix%
\section{Left-hardness using the box complex}\label{app:topo}
\subsection*{Basic definitions in topology}
For topological spaces $X,Y$, we call a continuous function $f: X \to Y$ a \emph{map}, for short. 
Two maps $f,g: X \to Y$ are \emph{homotopic} if they can be continuously transformed into one another; formally: there is a family of maps $\phi_t: X \to Y$ for $t\in [0,1]$ (called a \emph{homotopy}) such that $\phi_0 = f$, $\phi_1=g$ and such that the function $(t,x) \mapsto \phi_t(x)$ from $[0,1] \times X$ to $Y$ is continuous.
Two spaces $X,Y$ are \emph{homotopy equivalent} if there are maps $f:X \to Y$ and $g: Y \to X$ such that $g \circ f$ and $f \circ g$ are homotopic to identity maps on $X$ and on $Y$.

We shall only consider topological spaces described in the following simple combinatorial way.
A \emph{(simplicial) complex} $K$ is a family of non-empty finite sets that is downward closed, in the sense that $\emptyset \neq \sigma'\subseteq \sigma \in K$ implies $\sigma' \in K$.
The sets in $K$ are called \emph{faces} (or \emph{simplices}) of the complex, while their elements $V(K) := \bigcup_{\sigma \in K} \sigma$ are the \emph{vertices} of the complex.
The \emph{geometric realisation} $|\sigma|$ of a face $\sigma \in K$ is the subset of $\RR^{V(K)}$ defined as the convex hull of $\{e_v \mid v \in \sigma\}$, where $e_v$ is the standard basis vector corresponding to the $v$ coordinate in $\RR^{V(K)}$.
The geometric realisation $|K|$ of $K$ is the topological space obtained as the subspace $\bigcup_{\sigma \in K} |\sigma| \subseteq \RR^{V(K)}$.
We represent the points of $|K|$ as linear combinations of vertices $\lambda_1 v_1 + \dots \lambda_n v_n$ such that $\{v_1,\dots,v_n\}\in K$ and $\lambda_i$ are non-negative reals summing to 1.
We often refer to $K$ itself as a topological space, meaning $|K|$.
A \emph{simplicial map} $K\to K'$ is a function $f \colon V(K) \to V(K')$ such that $f(\sigma) \defeq \{f(v) \mid v\in \sigma\}$ is a face of $K'$ whenever $\sigma$ is a face of $K$.
It induces a map $|f|\colon|K|\to |K'|$ by extending it linearly from vertices on each face: $|f|(\sum_i \lambda_i v_i) \defeq \sum \lambda_i f(v_i)$.

For example, the circle may be represented as the triangle $K=\{\{1\},\{2\},\{3\}, $ $\{1,2\}, \{2,3\}, \{3,1\}\}$, meaning that $|K|$, which is the sum of three intervals in $\RR^3$, is homotopy equivalent to the unit circle $\Sphere^1$ in $\RR^2$. Adding the face $\{1,2,3\}$ to $K$ would make $|K|$ contractible, that is, homotopy equivalent to the one-point space.

\subsection*{Equivariant topology -- topology with symmetries}
Rather than asking about ``non-trivial maps'' (maps not homotopic to a constant map) it is easier to work with equivariant topology, that is, considering topological spaces together with their symmetries and symmetry-preserving maps.
A \emph{$\ZZ_2$-space} is a topological space $X$ equipped with a map $-: X \to X$, called a $\ZZ_2$-action on $X$, satisfying $-(-x) = x$ (for all $x\in X$).
We will call $-x$ the \emph{antipode} of $x$.
The main example is the $n$-dimensional \emph{sphere}: the $\ZZ_2$-space defined as the unit sphere in $\RR^{n+1}$ with $\ZZ_2$-action $x \mapsto -x$ as vectors.
A \emph{$\ZZ_2$-map} from $(X,-_X)$ to $(Y,-_Y)$ is a map $f: X \to Y$ that preserves the symmetry: $f(-_X\ x) = -_Y\ f(x)$ (this is also called an \emph{equivariant} map).
We write $X \to_{\ZZ_2} Y$ if such a map exists (the $\ZZ_2$-actions being clear from context).

Standard notions extend in a fairly straightforward way to equivariant notions.
A \emph{$\ZZ_2$-complex} is a simplicial complex $K$ together with a function $- : V(K) \to V(K)$ such that $-(-v) = v$ and $-\sigma \defeq \{-v \mid v \in \sigma\}\in K$ for $\sigma \in K$; this induces a $\ZZ_2$-action on $|K|$.
The \emph{product} of two $\ZZ_2$-spaces $X,Y$ is $X \times Y$ with ``simultaneous'' $\ZZ_2$-action $(x,y)\mapsto (-x,-y)$.
A homotopy $\phi_t$ between $\ZZ_2$-maps $f,g:X \to Y$ is called a \emph{$\ZZ_2$-homotopy} if $\phi_t$ is a $\ZZ_2$-map for all $t\in[0,1]$.
We say that two $\ZZ_2$-spaces $X,Y$ are $\ZZ_2$-homotopy equivalent, denoted $X \zeq Y$, if there are $\ZZ_2$-maps $f: X \to_{\ZZ_2} Y$ and $g: Y \to_{\ZZ_2} X$ such that $g \circ f$ and $f \circ g$ are $\ZZ_2$-homotopic to the identity.
Note this is stronger than just requiring $X \to_{\ZZ_2} Y$ and $Y \to_{\ZZ_2} X$; homotopy equivalence is more similar to graph isomorphism than to homomorphic equivalence of graphs.

\subsection*{The box complex -- the topology of a graph}
The \emph{box complex} $\Bx{G}$ of a graph $G$ is a $\ZZ_2$-complex defined as the family of vertex sets of complete bipartite subgraphs of $G \times K_2$ (with both sides non-empty) and their subsets.
In particular it contains all edges of $G \times K_2$ and every $K_{2,2} = C_4$ subgraph.
The topology of box complexes of the following graphs is folklore.

\begin{lemma}\label{lem:boxes}
	The following spaces are $\ZZ_2$-homotopy equivalent:
	\begin{enumerate}[label=(\roman*),itemsep=0pt]
	\item $\BBox{K_n} \zeq \Sphere^{n-2}$ for $n\geq 2$,
	\item $\BBox{C_n} \simeq_{\ZZ_2} \Sphere^1$ for odd $n\geq 3$,
	\item $\BBox{K_{p/q}} \simeq_{\ZZ_2} \Sphere^1$ for $2<\frac{p}{q}<4$,
	\item for every loop-less square-free graph $K$, $\BBox{K}$ is $\ZZ_2$-homotopy-equivalent to a 1-dimensional complex (a complex in which every face has at most 2 vertices).
	\end{enumerate}
\end{lemma}
\begin{proof}
	For (i), see Proposition 19.8 in~\cite{Kozlov2008book}, Proposition 4.3 in~\cite{babson2006}, or Lemma 5.9.2 in~\cite{matousek2008using}.
	Informally, the vertices of $\Bx{K_n}$ can be mapped bijectively to points in $\RR^n$ of the form $\pm e_i \defeq (0,\dots,0,\pm1,0,\dots,0)$.
	These are vertices of the cross-polytope in $\RR^n$ (the $n$-dimensional counterpart of the octahedron).
	Faces of $\Bx{K_n}$ are exactly those subsets of $\{\pm e_1,\dots,\pm e_n\}$ that do not contain repeated indices ($+e_i$ and $-e_i$ for any $i$), except for the two sets $\{+e_1,\dots,+e_n\}$ and $\{-e_1,\dots,-e_n\}$ (since a bipartite complete graph containing all $n$ vertices on one side cannot contain any vertex on the other side).
	The complex is thus isomorphic to the cross-polytope (the $n$-dimensional counterpart to the octahedron) in $\RR^n$, but with the interior and two opposite facets removed.
	The cross-polytope after removing the interior is $\ZZ_2$-homotopy equivalent to $\Sphere^{n-1}$ and after removing two opposite facets it is $\ZZ_2$-homotopy equivalent to $\Sphere^{n-2}$.
	
	For (iv), let use denote the two vertices of $\Bx{K}$ corresponding to $v \in V(K)$ as $\vL{v}$ and $\vR{v}$.
	Observe that $\Bx{K}$ would be isomorphic to $K \times K_2$ (meaning the 1-dimensional simplicial complex with $V(K \times K_2)$ as vertices and with $E(K \times K_2)$ and their subsets as faces), except that it also contains $N[\vL{v}] \defeq \{\vL{v}\} \cup \{\vR{w} \colon w \in N(v)\}$ and $N[\vR{v}]$ for each $v \in V(K)$ (except those with empty neighbourhood).
	However, these additional faces can be collapsed.
	Formally, every face not in $E(K \times K_2)$ is either of the form $\{\vL{v},\vR{w_1},\dots,\vR{w_n}\}$ or $\{\vR{w_1},\dots,\vR{w_n}\}$ for some $w_i \in N(v)$ and $n\geq 2$,
	or the same with $\circ$ and $\bullet$ swapped.
	Since $K$ is square-free, even in the second case $v$ is uniquely determined by the $w_i$.
	Hence we can match these faces in pairs.
	This matching is easily checked to satisfy the definitions of a so-called \emph{acyclic $\ZZ_2$-matching} in Discrete Morse Theory,
	which allows to show that removing these faces gives a $\ZZ_2$-homotopy equivalent complex:
	see Section 3 in~\cite{Wrochna17b} for definitions and details.

	For (ii), observe that by the above, $\Bx{C_n}$ is $\ZZ_2$-homotopy equivalent to $C_n \times K_2 = C_{2n}$ as a simplicial complex (for odd $n$).
	It is straightforward to give a $\ZZ_2$-homotopy equivalence (in fact a homeomorphism) to $\Sphere^1$.

	For (iii), we first consider the case when $p$ is odd.
	Then, $K_{p/q} \times K_2$ is isomorphic to the Caley graph $K'$ of $\ZZ_{2p}$ with generators $\{\pm 1,\pm 3,\dots,\pm p-2q\}$ (the isomorphism maps $(i,0)$ to $2i$ and $(i,1)$ to $2i+p$).
	In particular, $K'$ includes a cycle $C_{2p}$ on $0,1,\dots,2p-1$
	and the $\ZZ_2$-action on $K_{p/q} \times K_2$ correspond to point reflection on $C_{2p}$.
	We thus have an inclusion map $\iota\colon |C_{2p}| \to |K'|$ (where $|K'|$ is is shorthand for $\BBox{K_{p/q}\times K_2}$ and $C_{2p}$ is meant as a subcomplex).
	Note that $\frac{p}{q} < 4$ is equivalent to $p-2q < \frac{p}{2}$, so two adjacent vertices of $K'$ are at distance at $<\frac{p}{2}$ in $C_{2p}$.
	Therefore, every face of the box complex (a complete bipartite subgraph of $K'$) is contained in an interval of length $<p$ in $\ZZ_{2p}$.
	Every point in the geometric realization of such a face can be unambiguously mapped by linear extension in the interval
	to a point in the geometric realization of $C_{2p}$, giving a $\ZZ_2$-map $f \colon |K'| \to |C_{2p}|$.
	The maps $\iota,f$ give a $\ZZ_2$-homotopy equivalence ($f \circ \iota \colon |C_{2p}|\to |C_{2p}|$ is equal to the identity, while $\iota \circ f$ is $\ZZ_2$-homotopic to the identity, since one can also linearly extrapolate between the definition of $f$ and the identity map).
	The proof for even $p$ is similar, the main difference being that $K'$ should be the graph on $\ZZ_p \times \{0,1\}$ with $(i,a)$ adjacent to $(j,b)$ if $a\neq b$ and $i,j$ are at distance $\leq\frac{p-2q}{2}$.
\end{proof}

Note that for a loop-less graph $K$, $\Bx{K}$ is a free $\ZZ_2$-complex, which means every face $\sigma$ is disjoint from $-\sigma$.
This in turn implies that $\BBox{K}$ is a free $\ZZ_2$-space, which means that a point is never its own antipode.
Proposition 5.3.2.(v) in~\cite{matousek2008using} shows that a free $\ZZ_2$-complex of dimension $n$ admits a $\ZZ_2$-map to $\Sphere^n$.
Hence for loop-less, square-free graphs $K$, we have $\BBox{K} \to_{\ZZ_2} \Sphere^1$.

\subsection*{The hom complex -- preserving products}
Instead, we will use the Hom complex $\Hom(K_2,G)$, which is $\ZZ_2$-homotopy equivalent to $\Bx{G}$,
as proved by Csorba~\cite{Csorba08}.
Its vertices are homomorphisms $K_2 \to G$, that is, oriented edges $(u,v)$ of $G$.
For every $U,V \subseteq V(G)$ such that $U \times V \subseteq E(G)$, $U \times V$ and its subsets are faces of $\Hom(K_2,G)$.
In other words, a set $\sigma$ of oriented edges is a face if for every two $(u,v), (u',v') \in \sigma$, $(u,v')$ is an oriented edge of $G$.
The $\ZZ_2$-action swaps $(u,v)$ to $(v,u)$.

This definition has the advantage that it respects products trivially (and exactly, not just up to homotopy equivalence):
$\Hom(K_2,G \times H)$ is isomorphic to $\Hom(K_2,G) \times \Hom(K_2,H)$ (as $\ZZ_2$-simplicial complexes).
The isomorphism simply maps the oriented edge between pairs $(g_1,h_1)$ and $(g_2,h_2) \in V(G) \times V(H)$ to the pair of oriented edges $((g_1,h_1),(g_2,h_2))$.
In the same way, $\Hom(K_2,G^L)$ is isomorphic $\Hom(K_2,G)^L$, mapping pairs of
$L$-tuples to $L$-tuples of pairs.

\begin{lemma}\label{lem:minionComplex}
	Let $f \colon G^L \to H$ be a graph homomorphism.
	Let $f' \colon \Hom(K_2,G)^L \to \Hom(K_2,H)$ be the induced simplical $\ZZ_2$-map, defined as:
	$$ f'((u_1,v_1),\dots,(u_L,v_L)) \defeq (f(u_1,\dots,u_L),f(v_1,\dots,v_L)).$$
	Then the transformation $f \mapsto f'$ preserves minors and composition.
\end{lemma}
This is straightforward from the definitions.
Here by compositions we mean functions of the form $h(f(g_1(x_1),\dots,g_L(x_L)))$ for $g_i\colon G' \to G$ and $h\colon H \to H'$; the graph homomorphisms $g_i$ and $h$ induce simplicial maps just as above for $L=1$.
Preserving compositions means in particular that if $\mu$ is an automorphism of $G$ and $\mu'$ is the automorphism of $\Hom(K_2,G)$ it induces, then $f(x_1,\dots,\mu(x_i),\dots,x_L)$ induces $f'(x_1,\dots,\mu'(x_i),\dots,x_L)$).

In the geometric realisation, the above-mentioned isomorphism induces (by linear extension) an isomorphism from $\left|\Hom(K_2,G \times H)\right|$
to $\left|\Hom(K_2,G) \times \Hom(K_2,H)\right|$.
The latter has a natural $\ZZ_2$-homotopy equivalence to $\left|\Hom(K_2,G)\right| \times \left|\Hom(K_2,H)\right|$, implicit in the following claim:

\begin{lemma}\label{lem:minionHTop}
	Let $f \colon X^L \to Y$ be a $\ZZ_2$-simplicial map and let $x_0 \in V(X)$.
	Let $|f| \colon |X|^L \to |Y|$ be the induced $\ZZ_2$-map, defined as:
	$$\textstyle|f|(\sum_i \lambda^{(1)}_i v^{(1)}_i, \dots, \sum_i \lambda^{(L)}_i v^{(L)}_i) \defeq \sum_{i_1,\dots,i_L} \lambda_{i_1}\cdots \lambda_{i_L} f(v^{(1)}_{i_1},\dots,v^{(L)}_{i_L})$$
	for faces $\{v^{(1)}_i \mid i\}, \dots, \{v^{(L)}_i \mid i\} \in X$.
	Then the transformation $f \mapsto |f|$ preserves minors up to $\ZZ_2$-homotopy rel $x_0$ and preserves composition exactly.
\end{lemma}
\begin{proof}
	Preservation of composition is again straightforward.
	
	To see that the transformation preserves minors, consider for example the contraction (identification) of two coordinates.
	The general case is entirely analogous.
	Let $f \colon X^2 \to Y$ and let $f_{/2} \colon X \to Y$ be the minor obtained by contracting the two coordinates.
	Then 
	$$\textstyle |f_{/2}|(\sum_i \lambda_i v_i) = \sum_i \lambda_i f_{/2}(v_i) = \sum_i \lambda_i f(v_i, v_i).$$
	On the other hand, if we take the induced map first and only then contract, we obtain:
	$$\textstyle |f|_{/2}(\sum_i \lambda_i v_i) = |f|(\sum_i \lambda_i v_i, \sum_i \lambda_i v_i) = \sum_{i,j} \lambda_i \lambda_j f(v_i, v_j).$$
	The first point is in the face $\{f(v_i,v_i) \mid i\}$ of $Y$,
	the second is in the face $\{f(v_i,v_j) \mid i,j\}$ of $Y$ which contains the former.
	We can thus continuously move from one to the other.
	Formally, let $\mu_{i,j} \defeq \lambda_i$ if $i=j$ and $0$ otherwise.
	Then the functions (for $t \in [0,1]$)
	$$\textstyle f_t(\sum_i \lambda_i v_i) \defeq \left(t \cdot \mu_{i,j} + (1-t)\cdot \lambda_i \lambda_j\right) f(v_i,v_j)$$
	are always well-defined and give a $\ZZ_2$-homotopy between $|f_{/2}|$ and $|f|_{/2}$.
	For any vertex $x_0$ (i.e. $\lambda_1=1$) $f_t(x_0)$ is constantly equal to $f(x_0)$.
\end{proof}

We thus have a minion homomorphism from $\Pol(G,H)$ to the minion of maps-up-to-homotopy $|\Hom(K_2,G)|^L \to |\Hom(K_2,H)|$, which preserves automorphisms of~$G$.
This, as well as the minion homomorphism in the following subsection, can be interpreted as an instance of Lemma~\ref{lem:cat2}.

\subsection*{The fundamental group}
For a topological space $|X|$ and a point $x_0 \in |X|$, two maps from $|X|$ to some topological space are \emph{homotopic rel $x_0$} if there are homotopies that do not move the image of $x_0$.
In the \emph{fundamental group} $\pi_1(|X|, x_0)$, the elements are equivalence classes of loops at $x_0$ (maps $[0,1] \to |X|$ mapping 0 and 1 to $x_0$) under homotopy rel $x_0$,
the group operation is concatenation.
We skip $x_0$ when it is not important, since $\pi_1(|X|,x_0)$ is always
isomorphic to $\pi_1(|X|,x_0')$ if $|X|$ is path-connected\footnote{All the spaces we consider come from finite simplicial complexes, so connectivity in the topological sense is equivalent to path-connectivity (every two points being connected by a path) and to connectivity of the complex (as in a graph).} which we implicitly assume throughout.

Including information about the $\ZZ_2$-symmetry in the fundamental group is a bit less obvious.
For a $\ZZ_2$-space $|X|$ we can look at the fundamental group of $|X|$ but also the fundamental group of the quotient $|X|_{/\ZZ_2}$ (where every point is identified with its antipode; a.k.a. the orbit space or base space; we denote the equivalence class of $x$ by $\pm x$).
One way to think of elements of $\pi_1(|X|_{/\ZZ_2}, \pm x_0)$ is as paths from $x_0$ to either $x_0$ or~$-x_0$, with concatenation defined using the $\ZZ_2$-action if necessary.
Observe that $\pi_1(|X|_{/\ZZ_2})$ contains $\pi_1(|X|)$ as a subgroup, consisting of paths from $x_0$ to $x_0$.

Another way to describe the subgroup is by a group homomorphism to $\nu_X \colon \pi_1(|X|_{/\ZZ_2}) \to \ZZ_2$ mapping the subgroup (paths $x_0$ to $x_0$) to 0 and everything else (paths $x_0$ to $-x_0$) to~1.
Thus $\pi_1(|X|)$ is the subgroup given by the kernel of $\nu_X$.%
\footnote{In group theory, one would say $\pi_1(|X|)$ is a normal subgroup of index 2, or that $\pi_1(|X|)\to \pi_1(|X|_{/|\ZZ_2|}) \to \ZZ_2$ is a short exact sequence. In topology, one would say that $|X|$ is a degree-2 covering, or double cover, of $|X|_{/\ZZ_2}$; the group homomorphism $\nu_X$ is the \emph{monodromy action}, acting on the set $\{x_0,-x_0\}$.}

For example, consider $\Sphere^1$. The quotient $\Sphere^1_{/\ZZ_2}$ is again a circle,
so $\pi_1(\Sphere^1_{/\ZZ_2})$ is isomorphic to $\ZZ$ (a loop in the quotient is represented by its winding number);
$\nu$ is the remainder mod 2 (loops with odd winding number in the quotient correspond to paths from a point to its antipode in $\Sphere^1$)
and $\pi_1(\Sphere^1)$ is the subgroup $2\ZZ$ of even integers.
In contrast, the quotient $\Sphere^2_{/\ZZ_2}$ is the projective plane, so $\pi_1(\Sphere^2_{/\ZZ_2})$ is isomorphic to $\ZZ_2$; $\nu$ is the identity and the subgroup $\pi_1(\Sphere^2)$ is the trivial group.

A map $f\colon |X| \to |Y|$ induces a group homomorphism $f_* \colon \pi_1(|X|_{/\ZZ_2}) \to \pi_1(|Y|_{/\ZZ_2})$, simply by composing a loop with $f$.
This homomorphism preserves the subgroup: $\nu_Y(f_*(x))=\nu_X(x)$.
Equivalently, $f_*^{\ -1}(\pi_1(|Y|))=\pi_1(|X|)$.

The fundamental group of a product $\pi_1(|X|\times|Y|,(x_0,y_0))$ is isomorphic to the direct product of fundamental groups  $\pi_1(|X|,x_0) \times \pi_1(|Y|,y_0)$.
The isomorphism just maps a loop $\Sphere^1 \to |X|\times|Y|$ to the pair of loops obtained by composing with projections; the inverse maps a pair of loops $p \colon \Sphere^1 \to |X|$ and $q \colon \Sphere^1 \to |Y|$ to the ``simultaneous'' loop $(p,q) \colon t \mapsto (p(t),q(t))$.

However, $\pi_1((|X|\times|Y|)_{/\ZZ_2})$ is not isomorphic to $\pi_1(|X|_{/\ZZ_2}) \times \pi_1(|Y|_{/\ZZ_2})$, but to the subgroup of it given by elements $(x,y)$ such that $\nu_X(x)=\nu_Y(y)$.
Indeed, it contains paths from $(x_0,y_0)$ to either $(x_0,y_0)$ or $(-x_0,-y_0)$ but not to $(x_0,-y_0)$.

In other words, to a $\ZZ_2$-space $|X|$ we assign a group $\pi_1(|X|_{/\ZZ_2})$ together with a group homomorphism $\nu_X$ to $\ZZ_2$.
Consider the category whose objects are such pairs $(G,\nu)$ (a group with a homomorphism to $\ZZ_2$),
while morphisms $(G,\nu_G) \to (H,\nu_H)$ are group homomorphisms $G\to H$ preserving $\nu$.
The categorical product of $(G,\nu_G)$ and $(H,\nu_H)$ is $\{(g,h) \in G \times H \colon \nu_G(g)=\nu_H(h)\}$ with coordinate-wise multiplication and the homomorphism to $\ZZ_2$ defined in an obvious way ($\nu(g,h)\defeq\nu_G(g)=\nu_H(h)$).
Let us denote this product as $\star$ for clarity\footnote{In category theory, $(G,\nu_G) \star (H,\nu_H)$ is called the \emph{pullback} of $\nu_G$ and $\nu_H$, and may be denoted $G \times_{\ZZ_2} H$.} and the $L$-fold product of $(G,\nu_G)$ as $G^{\star L}$.
Then a $\ZZ_2$-map $f \colon |X|^L\to|Y|$ mapping a point $x_0$ to $y_0$ induces a morphism $f_* \colon \pi_1(|X|_{/\ZZ_2}, \pm x_0)^{\star L} \to \pi_1(|Y|_{/\ZZ_2}, \pm y_0)$ (a group homomorphism preserving $\nu$):
	$$\textstyle f_*([p_1],\dots,[p_L]) \defeq [t \mapsto f(p_1(t),\dots,p_L(t))]$$
(where $[p]$ denotes the equivalence class of a loop $\Sphere^1 \to |X|_{/\ZZ_2}$ under homotopy rel $\pm x_0$).

The following is straightforward to check from the definition of $f_*$:

\begin{lemma}\label{lem:minionFundGroup}
	Let $f \colon |X|^L\to|Y|$ be $\ZZ_2$-map, let $x_0 \in |X|$ be an arbitrary point and let $f(x_0)=y_0$.
	Let $f_* \colon \pi_1(|X|_{/\ZZ_2})^{\star L} \to \pi_1(|Y|_{/\ZZ_2})$ be the induced group homomorphism.
	Then the transformation $f \mapsto f_*$ preserves minors and preserves automorphisms of $|X|$ that fix $x_0$.%
	\footnote{More generally, one could consider pointed spaces (pairs $(|X|,x_0)$) and pointed $\ZZ_2$-maps $(|X|,x_0) \to (|Y|,y_0)$ (maps that map $x_0$ to $y_0$).
	Then $f \mapsto f_*$ preserves composition with pointed maps; automorphisms fixing $x_0$ are a special case.}
\end{lemma}

\subsection*{Wrapping it up}
Let us denote $\pi_1(|\Hom(K_2,G)|)$ and $\pi_1(|\Hom(K_2,G)|/_{\ZZ_2})$ as respectively $\pi_1(G)$ and $\pi_1(G_{/\ZZ_2})$, for short.

Consider a graph homomorphism $f \colon C_n^L \to H$ ($n$ odd).
We have $|\Hom(K_2, C_n)| \simeq_{\ZZ_2} \Sphere^1$
and hence $\pi_1({C_n}_{/\ZZ_2})$ is $\ZZ$
with a group homomorphism $\nu_{C_n}\colon i \mapsto (i\mod 2)$.
In particular $\ZZ^{\star L}$ is the subgroup of $\ZZ^L$ given by $L$-tuples in which the integers are all even or all odd and $\pi_1(C_n)$ is the subgroup $2\ZZ$ of even integers in $\ZZ$.
For an arbitrarily fixed edge $e_0$ of $C_n$, the automorphism $\mu_{C_n}$ that mirrors the graph and fixes $e_0$
induces the automorphism of $\ZZ$ which maps $i$ to $-i$.

Therefore, composing the transformations from Lemmas~\ref{lem:minionComplex}, \ref{lem:minionHTop}, and \ref{lem:minionFundGroup},
we obtain a~group homomorphism $f_* \colon \ZZ^{\star L} \to \pi_1(H_{/\ZZ_2})$ which preserves the homomorphism to $\ZZ_2$ and the mirror automorphism on each coordinate.

Suppose that $|\Hom(K_2,H)|\zeq \Sphere_1$,
so again $\pi_1(H_{/\ZZ_2}) = \ZZ$ with the same homomorphism to $\ZZ_2$ ($i \mod 2$) and the same mirror automorphism ($-i$).
Since $f_*$ preserves the homomorphism to $\ZZ_2$, $d \defeq f_*(1,1,\dots,1) \in \ZZ$ is an odd number,
which means $f_*(2,2,\dots,2)=2d$ is non-zero.
This is why we needed the $\ZZ_2$-action: to conclude that $f_*$ is non-trivial.
We can now focus on what $f_*$ does on the subgroup of even integers.

Let $a_\ell \defeq f_*(0,\dots,0,2,0,\dots,0) \in  \ZZ$ with a $2$ in the $\ell$-th coordinate.
Then $f_*$ on even numbers is completely determined by these elements: $f_*(2i_1,\dots,2i_L) = a_1 \cdot {i_1} + \cdots + a_L \cdot i_L$ (because it is a group homomorphism).
By the above, $\sum_{\ell=1}^L a_\ell$ is non-zero.
Since $f \mapsto f_*$ preserves minors, 
we know that the minor $i \mapsto f_*(i,i,\dots,i)$ is a group homomorphism induced by some graph homomorphism $C_n \to H$ (namely by the corresponding minor $v \mapsto f(v,\dots,v)$), hence the integer $f_*(2,2,\dots,2)$ belongs to a set of at most $|H|^n$ possibilities.
The same holds for compositions with mirror symmetries:
the group homomorphism $i \mapsto f_*(i,\dots, -i, \dots,i)$ with a minus on any subset of coordinates is 
induced by the graph homomorphism $C_n \to H$ defined as
$f(v,\dots,\mu_{C_n}(v),\dots,v)$ with $\mu_{C_n}$ on the same set of coordinates.
Hence for $i_1,\dots,i_L \in \{+1,-1\}$, the values $f_*(2 i_1,2 i_2,\dots,2 i_L) = a_1 \cdot i_1 + \cdots + a_L \cdot i_L$ belong to a set of at most $|H|^n$ possibilities.
This implies less than $|H|^n$ of the integers $a_\ell$ are non-zero.
Indeed, if there are $L'$ coordinates $\ell$ for which $a_\ell$ is non-zero, then one can set the corresponding $i_\ell$ to make $a_\ell \cdot i_\ell$ positive, and then swap $i_\ell$ one-by-one in any order, resulting in a strictly decreasing sequence of values $a_1 \cdot i_1 + \cdots a_L \cdot i_L$, hence in $L'+1$ distinct values.
Hence $L'+1 \leq |H|^n$.

Therefore, the group homomorphism $(2i_1,\dots,2i_L) \mapsto
f_*(2i_1,\dots,2i_L) \colon (2\ZZ)^L \to (2\ZZ)$ has bounded (but non-zero)
essential arity.
Note that this is exactly the homomorphism $f_*|_{\pi_1(C_n)^L}$, from the subgroup $\pi_1(C_n)^L$ to the subgroup $\pi_1(H)$.
Therefore, the transformation $f \mapsto f_*|_{\pi_1(C_n)^L}$ is a minion homomorphism from $\Pol(C_n,H)$ to a minion of functions of bounded essential arity.

The same argument would work if instead of $|\Hom(K_2,H)|\zeq \Sphere_1$ we only assumed we had a $\ZZ_2$-map $g\colon|\Hom(K_2,H)|\to \Sphere_1$,
since it would induce a group homomorphism $g_* \colon \pi_1(H_{/\ZZ_2}) \to \ZZ$ which preserves the homomorphism to $\ZZ_2$, in a way that preserves mirror automorphisms of $\Sphere^1$; it then suffices to compose $g_*$ with $f_*$ and continue as above.

This concludes the proof of the following:
\begin{theorem}
	Let $H$ be a graph such that $|\Hom(K_2,H)| \to_{\ZZ_2} \Sphere_1$.
	Then for all odd $n$, $\Pol(C_n,H)$ admits a minion homomorphism to a minion of bounded essential arity with no constant functions.
\end{theorem}

\noindent
By Theorem~\ref{thm:minbndarity}, this concludes the direct proof that $\PCSP(C_n,H)$ is NP-hard for all odd $n$:
\begin{corollary}
	Let $H$ be a graph such that $|\Hom(K_2,H)| \to_{\ZZ_2} \Sphere_1$.
	Then $H$ is left-hard.
\end{corollary}

\noindent
Since $|\Hom(K_2,H)|$ is $\ZZ_2$-homotopy equivalent to $\BBox{H}$ (hence they admit the same $\ZZ_2$-maps),
this is exactly equivalent to Corollary~\ref{cor:S1};
in particular it gives a proof of Theorem~\ref{thm:K3}.

\subsection*{Further remarks}
In the case of $\Sphere^1$, the fact that a $\ZZ_2$-map $g\colon|\Hom(K_2,H)|\to \Sphere_1$ induces a group homomorphism $g_* \colon \pi_1(H_{/\ZZ_2}) \to \ZZ$ which preserves the homomorphism to $\ZZ_2$ is in fact an exact characterisation.
That is, as stated by Matsushita~\cite{Matsushita17}, standard covering space theory yields the following:
\begin{lemma}
	A connected $\ZZ_2$-space $|X|$ admits a $\ZZ_2$-map to $\Sphere^1$ if and only if there exists a group homomorphism $f \colon \pi_1(|X|_{/\ZZ_2}) \to \ZZ$ which preserves the action (that is, $f^{-1}(2\ZZ) = \pi_1(X)$).
\end{lemma}

In the above proof, one could go directly from graphs to fundamental groups, avoiding simplicial complexes and topological spaces (though they remain the simplest way to prove that these fundamental groups preserve products).
A direct definition of the fundamental group of the quotient space $\BBox{H}_{/\ZZ_2}$ is as follows.
We consider closed walks (cycles that are allowed to self-intersect) from an arbitrary fixed vertex $v_0 \in V(H)$.
Two such walks are consider equivalent if one can be obtained from the other by adding/removing backtracks (a pair of consecutive edges going back and forth on the same edge of $H$) and 4-cycles (subwalks around a cycle of length 4).
The elements of the group are equivalence classes of walks, with concatenation as multiplication.
The resulting group is isomorphic to $\pi_1(H_{/K_2})$ (this combinatorial definition is known as the \emph{edge-path group}; see~\cite{Matsushita12} or Section 3.6 and 3.7 in~\cite{Spanier}).
Considering walks in $H \times K_2$ instead would yield a group isomorphic to $\pi_1(H)$.

For example, for odd cycles and more generally circular cliques $<4$ the group is just $\ZZ$ (Lemma 4.1 in~\cite{Wrochna17} has a direct but technical proof), for square-free graphs the group is a free (non-Abelian) group.
For $K_4$, the resulting group is just $\ZZ_2$ (all walks of the same parity are equivalent), which corresponds to the fact that $\BBox{K_4}$ is the 2-sphere and $\BBox{K_4}_{/\ZZ_2}$ is the projective plane.

Unfortunately, this makes the fundamental group useless for the question of whether $K_4$ is left-hard.
Indeed, there is only one possible induced group homomorphism $f_* \colon \ZZ^{\star L} \to \pi_1({K_4}_{/\ZZ_2}) = \pi_1(\RR\mathrm{P}_2) = \ZZ_2$: it maps $L$-tuples of even integers to 0 and  $L$-tuples of odd integers to 1 (because it has to preserve the homomorphism to $\ZZ_2$, which is the identity).
Whether other tools of algebraic topology can be useful remains to be seen.

\printbibliography

\end{document}

%% file: figBox.tex
\tikzset{ %
	L/.style={circle,draw=black,fill=white,inner sep=0pt,minimum size=4pt},
	R/.style={circle,fill=black,inner sep=0pt,minimum size=4pt},
	v/.style={circle,draw=black!75,inner sep=0pt,minimum size=8pt},
	h/.style={draw=black!50},
}

\begin{tikzpicture}[scale=1.1]
\begin{scope}[shift={(-5.2,-1.3)}]
	\newcommand\zscale{1.3}
	\coordinate (z) at ($\zscale*(0.6,0.4)$);
	\fill[fill=blue,opacity=0.1] (0,2)--(2,2)--(2,0)--(0,0)--cycle; 
	\fill[fill=blue,opacity=0.2] (0,2)--($(0,2)+(z)$)--($(2,2)+(z)$)--(2,2)--cycle; 
	\fill[fill=blue,opacity=0.15] (2,2)--($(2,2)+(z)$)--($(2,0)+(z)$)--(2,0)--cycle; 
	\fill[fill=blue,opacity=0.1] (0,0)--($(0,0)+(z)$)--($(2,0)+(z)$)--(2,0)--cycle; 
	\fill[fill=blue,opacity=0.05] (0,0)--($(0,0)+(z)$)--($(0,2)+(z)$)--(0,2)--cycle; 
	\node[L,label={135:$1$}] (v1L) at (0,2) {};
	\node[R,label={135:$2$}] (v2R) at (2,2) {};
	\node[L,label={135:$3$}] (v3L) at (2,0) {};
	\node[R,label={135:$4$}] (v4R) at (0,0) {};
	\draw (v1L)--(v2R)--(v3L)--(v4R)--(v1L);
	\node[R,label={135:\small$1$}] (v1R) at ($(2,0)+(z)$) {};
	\node[L,label={135:\small$2$}] (v2L) at ($(0,0)+(z)$) {};
	\node[R,label={135:\small$3$}] (v3R) at ($(0,2)+(z)$) {};
	\node[L,label={135:\small$4$}] (v4L) at ($(2,2)+(z)$) {};
	\draw[h] (v1R)--(v2L)--(v3R)--(v4L)--(v1R);
	\draw (v1L)--(v3R) (v2R)--(v4L) (v3L)--(v1R) (v4R)--(v2L);
	\node at (1.2,-0.5) {$\BBox{K_4}$};
\end{scope}
\begin{scope}[shift={(0,0)}]
	\node at (-1.8,0) {\Large$\zeq$};
	\shade[ball color = blue!40, opacity = 0.4] (0,0) circle (1.3);
	\node at (0,-1.7) {$\Sphere^2$};
\end{scope}
\begin{scope}[shift={(3.4,0)}]
	\node at (-1.6,0) {\Large$\not\to_{\ZZ_2}$};
	\draw (0,0) circle (1.1);
	\node at (0,-1.7) {$\Sphere^1$};
	\node at (1.6,0) {\Large$\zeq$};
\end{scope}

\begin{scope}[shift={(7,0)}]
	\foreach \i in {0,...,6} 
	{
		\node (u\i) at (90-360*\i/7 : 1.3)  {};
		\node (v\i) at (-90-360*\i/7 : 1.3)  {};
	}
	\foreach \i in {0,...,6}  
	{
		\pgfmathtruncatemacro\im{mod(\i+6,7)};
		\pgfmathtruncatemacro\ipp{mod(\i+2,7)};
		\pgfmathtruncatemacro\ippp{mod(\i+3,7)};
		\fill[fill=blue,opacity=0.2] (u\i.center) -- (v\ippp.center) -- (u\im.center) -- (v\ipp.center) -- cycle;
		\fill[fill=blue,opacity=0.2] (v\i.center) -- (u\ippp.center) -- (v\im.center) -- (u\ipp.center) -- cycle;
	}
	\foreach \i in {0,...,6} 
	{
		\node[L,label={{90-360*\i/7}:\small$\i$}] at (u\i) {};
		\node[R,label={{-90-360*\i/7}:\small$\i$}] at (v\i) {};
	}
	\foreach \i in {0,...,6}  
	{
		\pgfmathtruncatemacro\im{mod(\i+6,7)};
		\pgfmathtruncatemacro\ipp{mod(\i+2,7)};
		\pgfmathtruncatemacro\ippp{mod(\i+3,7)};
		\draw (u\i)--(v\ipp) (u\i)--(v\ippp);
		\draw (v\i)--(u\ipp) (v\i)--(u\ippp);

	}
	\node at (0,-1.9) {$\BBox{K_{7/2}}$};
\end{scope}

\begin{scope}[shift={(-1.2,-3.2)}]
	\node (u1) at (10:1) {\small$4$};
	\node (u2) at (-80:1) {\small$1$};
	\node (u3) at (190:1) {\small$2$};
	\node (u4) at (100:1) {\small$3$};
	\draw (u1)--(u2) (u1)--(u3) (u1)--(u4) (u2)--(u3) (u2)--(u4) (u3)--(u4);
	\node at (1,-1) {$K_4$};
\end{scope}
\begin{scope}[shift={(1.8,-3.2)}]
	\node at (0,0) {\Large$\not\to$};
\end{scope}
\begin{scope}[shift={(4.8,-3.2)}]
	\foreach \i in {0,...,6} 
	{
		\node (u\i) at (90-360*\i/7 : 1.1)  {\small$\i$};
	}
	\foreach \i in {0,...,6} 
	{
		\pgfmathtruncatemacro\ipp{mod(\i+2,7)};
		\pgfmathtruncatemacro\ippp{mod(\i+3,7)};
		\draw (u\i)--(u\ipp) (u\i)--(u\ippp);
	}
	\node at (-1.3,-1) {$K_{7/2}$};
\end{scope}
\end{tikzpicture}